\documentclass[final]{siamltex}
\usepackage{amsmath}
\usepackage{amssymb}
\usepackage{amsmath}
\def\H{{\mathbf{H}}}

\newcommand{\p}{\partial}

\newcommand{\fl}[2]{\frac{#1}{#2}}

\newcommand{\nn}{\nonumber}

\newcommand{\vep}{\varepsilon}

\newcommand{\sg}{\sigma}

\renewcommand{\theequation}{\arabic{section}.\arabic{equation}}
\newcommand{\be}{\begin{equation}}
\newcommand{\ee}{\end{equation}}
\newcommand{\ba}{\begin{array}}
\newcommand{\ea}{\end{array}}
\newcommand{\bea}{\begin{eqnarray}}
\newcommand{\eea}{\end{eqnarray}}
\newcommand{\beas}{\begin{eqnarray*}}
\newcommand{\eeas}{\end{eqnarray*}}

\newtheorem{remark}{Remark}[section]

 \newcommand{\bx}{{\bf x} }
 \newcommand{\bn}{{\bf n} }
 \newcommand{\br}{{\bf r} }
\begin{document}

\title{Gross-Pitaevskii-Poisson equations  for dipolar Bose-Einstein condensate with
anisotropic confinement}
\author{Weizhu Bao\footnotemark[2]
\and \fbox{Naoufel Ben Abdallah}\footnotemark[3]
\and Yongyong Cai\footnotemark[4]
}
\date{}

\maketitle
\renewcommand{\thefootnote}{\fnsymbol{footnote}}

\footnotetext[2]{Department of Mathematics and Center for Computational Science and
Engineering, National University of Singapore,
Singapore 119076 (bao@math.nus.edu.sg).}
\footnotetext[3]{IMT, UMR CNRS 5219, Universit\'{e} Paul
Sabatier, 31062 Toulouse Cedex, France.}
\footnotetext[4]{Department of Mathematics, National University of Singapore,
Singapore 119076 (caiyongyong@nus.edu.sg).}

\renewcommand{\thefootnote}{\arabic{footnote}}

\begin{abstract} Ground states and dynamical properties of
dipolar Bose-Einstein condensate are analyzed based on the
 Gross-Pitaevskii-Poisson system (GPPS) and its dimension reduction models
under anisotropic confining potential. We begin with the
three-dimensional (3D) Gross-Pitaevskii-Poisson system and review
its quasi-2D approximate equations when the trap is strongly
confined in $z$-direction and quasi-1D approximate equations when
the trap is strongly confined in $x$-, $y$-directions. In fact, in
the quasi-2D equations, a fractional Poisson equation with the
operator $(-\Delta)^{1/2}$ is involved which brings significant
difficulties into the analysis. Existence and uniqueness as well as
nonexistence of the ground state under different parameter regimes
are established for the quasi-2D and quasi-1D equations.
Well-posedness of the Cauchy problem for both equations and finite
time blowup in 2D are analyzed. Finally, we rigorously prove the
convergence and linear convergence rate between the solutions of the
3D GPPS and its quasi-2D and quasi-1D approximate equations in weak
interaction regime.
\end{abstract}

\pagestyle{myheadings}
\thispagestyle{plain}
\markboth{WEIZHU BAO, \fbox{NAOUFEL BEN ABDALLAH} AND YONGYONG CAI}
{GPE FOR DIPOLAR BEC WITH ANISOTROPIC CONFINEMENT}

\section{Introduction}
\setcounter{equation}{0}

Quantum degenerate gases have received considerable  interests
both theoretically and experimentally, since the first observation of
Bose-Einstein condensate (BEC) with dilute bosonic gas in 1995.
The properties of these  ultracold dilute quantum gases are determined
by the short-range, isotropic contact interactions between the particles,
which have been studied extensively. For those particles with large
permanent magnetic or electric dipole moment, dipole-dipole interactions
are non-negligible, and the dipolar interactions are long-range and
anisotropic, different from contact interactions.  Due to these remarkable
properties of dipolar interactions, there have been great interests
to study dipolar BEC in the last decade.  In 2005, the first dipolar BEC
with $^{52}$Cr atoms was successfully realized in experiments at the
Stuttgart University \cite{Gri}.
 Very recently in 2011, a dipolar BEC  with $^{164}$Dy atoms, whose
dipole-dipole interaction is much stronger than that of  $^{52}$Cr,
has been performed in experiments at the
Stanford University \cite{Lum}.  These success of
experiments have renewed interests in theoretically studying dipolar BECs.

In this paper, we will consider the zero temperature mean-field model of dipolar BEC,
 the three-dimensional (3D)  Gross-Pitaevskii equation (GPE) with dipolar interaction
 in dimensionless form
\cite{Bao1,Carles,Santos,YY1,YY2} \be \label{ngpe1} i\p_t
\psi(\br,t)=\left[-\frac{1}{2}\nabla^2+V(\br)+\beta |\psi|^2+\lambda
\left(U_{\rm dip}\ast|\psi|^2\right)\right]\psi, \qquad \br=(\bx,z)\in{\Bbb
R}^3, \quad t>0,\ee where $t$ is time, $\bx=(x,y)\in\Bbb R^2$ and
$\br=(\bx,z)=(x,y,z)\in\Bbb R^3$ are
the Cartesian coordinates, $\psi=\psi(\br,t)$ is the dimensionless
complex-valued wave function, $V(\br)$ is a given real-valued trapping
potential in the experiments,
$\beta$ and $\lambda$ are dimensionless constants representing the
 contact interaction and dipolar interaction, respectively,
and $U_{\rm dip}(\br)$ is given as \be\label{kel} U_{\rm dip}(\br)=
\frac{3}{4\pi}\,\fl{1-3(\br\cdot \bf
n)^2/|\br|^2}{|\br|^3}=\frac{3}{4\pi}\,
\fl{1-3\cos^2(\theta)}{|\br|^3}, \qquad \br\in{\Bbb R}^3,\ee with
the dipolar axis $\bn=(n_1,n_2,n_3)\in\Bbb R^3$ satisfying
$|\bn|=\sqrt{n_1^2+n_2^2+n_3^3}=1$. Although the kernel $U_{\rm dip}$
is highly singular near the origin,
the convolution is well-defined for $\rho\in L^p(\Bbb R^3)$
with $U_{\rm dip}*\rho\in L^p(\Bbb R^3)$  ($p\in(1,\infty)$) \cite{Carles}.
In the context of BEC, the initial data is usually
normalized such that $\|\psi(\cdot,0)\|_{L^2}=1$.

Denote the
differential operators $\partial_{\bn}=\bn\cdot\nabla$ and
$\partial_{\bn\bn}=\partial_{\bn}\partial_{\bn}$, and  notice the
identity \cite{Bao1} \be \label{decop1} U_{\rm dip}(\br)=\fl{3}{4\pi
|\br|^3}\left(1-\frac{3(\br\cdot {\bf n})^2}{|\br|^2}\right) = -
\delta (\br)-3\p_{\bn\bn}\left( \frac{1}{4\pi |\br|}\right),\qquad
\br\in {\Bbb R}^3,\ee with $\delta$ being the Dirac distribution, we can
re-formulate  the GPE (\ref{ngpe1})  as the following
Gross-Pitaevskii-Poisson system (GPPS) \cite{Bao1,Bao2} \bea
\label{gpe} &&i \p_t
\psi(\br,t)=\left[-\fl{1}{2}\nabla^2+V(\br)+(\beta-\lambda)
|\psi|^2-3\lambda \p_{\bn\bn} \varphi \right]\psi, \quad \br\in{\Bbb
R}^3,
\quad t>0,  \\
\label{poisson}&&\qquad \nabla^2 \varphi(\br,t) =
-|\psi(\br,t)|^2,\qquad \br\in{\Bbb R}^3, \qquad
\lim\limits_{|\br|\to\infty}\varphi(\br,t)=0,\qquad  t\ge0,
 \eea
 The above GPPS in 3D conserves the {\sl mass}, or the {\sl normalization} condition,
 \be
\label{norm3d}\|\psi(\cdot,t)\|_{2}^2=\int_{{\Bbb
R}^3}|\psi(\br,t)|^2\;d\br\equiv \int_{{\Bbb
R}^3}|\psi(\br,0)|^2\;d\br=1, \qquad t\ge0,\ee
 and  {\sl energy} per particle
\begin{eqnarray}
E_{3D}(\psi)=\int_{\Bbb
R^3}\left[\frac12|\nabla\psi|^2+V(\br)|\psi|^2+\frac{\beta-\lambda}{2}
|\psi|^4+\frac{3\lambda}{2}\left|\p_{\bn}\nabla\varphi\right|^2\right]\,d\br,
\quad \varphi=\frac{1}{4\pi|\br|}*|\psi|^2.
\end{eqnarray}
It was proven \cite{Bao1} that when $\beta\ge0$ and
$-\frac{\beta}{2}\le \lambda\le \beta$, there exists a unique
positive ground state $\Phi_g$ which is defined as the minimizer of
the energy functional, i.e. $E_{3D}(\Phi_g)=\min_{\|\Phi\|_2=1}\
E_{3D}(\Phi)$ and the Cauchy problem of the GPPS
(\ref{gpe})-(\ref{poisson}) is globally well-posed; otherwise
there exists no ground state and  the Cauchy problem is locally
well-posed and finite time blow-up may happen under certain
conditions \cite{Bao1}.

In many physical experiments of dipolar BECs,  the condensates
 are confined with  strong harmonic trap in one  or two axes directions,
 resulting in a pancake- or cigar-shaped dipolar BEC, respectively.
 Mathematically speaking, this corresponds to the anisotropic potentials
 $V(\br)$ of the form:

{\sl Case I} (pancake-shaped), potential is strongly confined in the
vertical $z$ direction
with \be\label{case1} V(\br)=V_2(\bx)+\frac{z^2}{2\vep^4}, \qquad
\br=(\bx,z)\in{\Bbb R}^3, \ee

 {\sl Case II} (cigar-shaped), potential is strongly confined in the
 horizontal $\bx=(x,y)\in {\Bbb R}^2$
 plane with
\be\label{case2} V(\br)=V_1(z)+\frac{x^2+y^2}{2\vep^4}, \qquad
\br=(\bx,z)\in {\Bbb R}^3, \ee
where $0<\vep\ll 1$ is a small parameter describing
the strength of confinement. In such cases, the above GPPS in 3D can be formally reduced
to 2D and 1D, respectively \cite{Bao2}.

In  {\sl Case I}, when $\vep\to 0^+$,  evolution of the solution
$\psi(\br,t)$ of GPPS (\ref{gpe})-(\ref{poisson}) in $z$-direction
would essentially occur in the ground state mode of
$L_z:=-\frac12\p_{zz}+\frac{z^2}{2\vep^4}$, which is spanned by
$w_\vep(z)=\vep^{-1/2}\pi^{-1/4}e^{-\frac{z^2}{2\vep^2}}$
\cite{Bao2,Bao21}. By taking the ansatz
\be\label{an1}\psi(\bx,z,t)=e^{-it/2\vep^2}\phi(\bx,t)w_\vep(z),
\qquad (\bx,z)\in{\Bbb R}^3, \quad t\ge0, \ee the 3D GPPS
(\ref{gpe})-(\ref{poisson}) will be formally reduced to a {\sl
{quasi-2D equation} I} \cite{Bao2}: \bea \label{gpe2d} i\p_t
\phi=\left[-\frac12\Delta+V_2+\frac{\beta-\lambda+3\lambda
n_3^2}{\sqrt{2\pi}\,\vep} |\phi|^2-\frac{3\lambda}{2}(
\p_{\bn_\perp\bn_\perp} -n_3^2\Delta)\varphi^{2D} \right]\phi, \quad
\bx\in {\Bbb R}^2, \ t>0, \qquad \eea where $\bx=(x,y)$,
$\bn_\perp=(n_1,n_2)$, $\p_{\bn_\perp}=\bn_\perp\cdot\nabla$,
$\p_{\bn_\perp\bn_\perp}=\p_{\bn_\perp}(\p_{\bn_\perp})$,
$\Delta=\p_{xx}+\p_{yy}$ and
 \be\label{u2d1}
 \varphi^{2D}(\bx,t)=U^{2D}_\vep*|\phi|^2,\quad
 U^{2D}_\vep(\bx)=\frac{1}{2\sqrt{2}\pi^{3/2}}
 \int_{\Bbb R}\frac{e^{-s^2/2}}{\sqrt{x^2+y^2+
\vep^2s^2}}\,ds, \quad \bx\in{\Bbb R}^2,\ t\ge0. \ee In addition, as
$\vep\to0^+$, $\varphi^{2D}$ can be approximated by
$\varphi^{2D}_\infty$ \cite{Bao2} as :
\be\label{u2d2} \varphi^{2D}_\infty(\bx,t)=U_{\rm dip}^{2D}*|\phi|^2, \quad
\hbox{with} \quad U_{\rm
dip}^{2D}(\bx)=\frac{1}{2\pi\sqrt{x^2+y^2}}, \qquad \bx\in{\Bbb
R}^2, \quad t\ge0, \ee which can be re-written as a fractional
Poisson equation \cite{Bao2} \be\label{u2d2k}
(-\Delta)^{1/2}\varphi^{2D}_\infty(\bx,t)=|\phi(\bx,t)|^2,\quad \bx\in{\Bbb
R}^2, \quad \lim\limits_{|\bx|\to\infty}\varphi^{2D}_\infty(\bx,t)=0,
\qquad t\ge0. \ee Thus an alternative {\sl{quasi-2D equation} II}
can be obtained as \cite{Bao2}: \be \label{gpe2d2} i \p_t
\phi=\left[-\frac12\Delta+V_2+\frac{\beta-\lambda+3\lambda
n_3^2}{\sqrt{2\pi}\,\vep} |\phi|^2-\frac{3\lambda}{2}(
\p_{\bn_\perp\bn_\perp} -n_3^2\Delta)(-\Delta)^{-1/2}(|\phi|^2)
\right]\phi.\ee

Similarly, in  {\sl Case II}, evolution of the solution
$\psi(\bx,z,t)$ of GPPS (\ref{gpe})-(\ref{poisson}) in $\bx=(x,y)$-directions
would essentially occur in the ground state mode of
$L_{\bx}:=-\frac12(\p_{xx}+\p_{yy})+\frac{x^2+y^2}{2\vep^4}$, which
is spanned by
$w_\vep(\bx)=\vep^{-1}\pi^{-1/2}e^{-\frac{|\bx|^2}{2\vep^2}}$
\cite{Bao2,Bao21}. Again, by taking the ansatz
\be\label{an2}\psi(\bx,z,t)=e^{-i t/\vep^2}\phi(z,t)w_\vep(\bx),
\qquad (\bx,z)\in{\Bbb R}^3, \quad t\ge0,\ee the 3D GPPS
(\ref{gpe})-(\ref{poisson}) will be formally reduced to a {\sl
{quasi-1D equation}} \cite{Bao2}:
 \bea \label{gpe1d} i \p_t
\phi=\left[-\frac{1}{2}\p_{zz}+V_1+\frac{\beta+\frac{1}{2}\lambda
(1-3n_3^2)}{2\pi\vep^2} |\phi|^2 -\frac{3\lambda(
3n_3^2-1)}{8\sqrt{2\pi}\,\vep}\p_{zz}\varphi^{1D} \right]\phi, \quad
z\in{\Bbb R}, \ t>0,\quad \eea where \be \label{poisson1d}
\varphi^{1D}(z,t)=U_\vep^{1D}*|\phi|^2, \qquad
U^{1D}_\vep(z)=\frac{\sqrt{2}e^{z^2/2\vep^2}}{\sqrt{\pi}\,\vep}\int_{|z|}^\infty
e^{-s^2/2\vep^2}\,ds, \qquad z\in{\Bbb R}, \quad t\ge0.
 \ee

The above effective lower dimensional models in 2D and 1D are very
useful in the study of dipolar BEC since  they are much easier and
cheaper to be simulated in practical computation. In fact, for the
GPE without the dipolar term, i.e. $\lambda=0$, there
 have been extensive  studies on this subject. For formal analysis
 and  numerical simulation, the convergence
 rate of such dimension reduction was investigated numerically in
 \cite{Bao3,Bao4} and a nonlinear Schr\"{o}dinger equation  with
 polynomial nonlinearity in reduced dimensions was proposed in \cite{Luca}.
 For rigorous analysis,
 convergence of the  dimension reduction under anisotropic confinement
 has been proven  in the weak
 interaction regime \cite{bacm,bamsw}, i.e.
 $\beta=O(\vep)$ in 2D  and $\beta=O(\vep^2)$ in 1D.
 However, with the dipolar term, i.e. $\lambda\ne0$, there were few works  towards
 the mathematical analysis for this dimension reduction except some preliminary results in
 \cite{Carles} where different scalings and formaulation were adapted. In fact,
 our quasi-2D models (\ref{gpe2d}) and (\ref{gpe2d2}) and quais-1D model
 (\ref{gpe1d}) are much easier
 to be used in mathematical analysis and practical numerical computation.



The main aim of this paper is to establish existence and uniqueness
of the ground states and well-posedness of the Cauchy problems
associated to the quasi-2D equations I and  II and quasi-1D
equation, and to analyze the convergence and convergence rate of the
dimension reduction from 3D to 2D and 1D. In order to do so, without
loss of generality, we assume the potential $V_d(\eta)\ge0$
 for $\eta\in{\Bbb R}^d$ ($d=1,2,3$).
 It is natural to consider the energy
space in  $d$-dimensions ($d=1,2,3$) defined as
$$ X_d=\left\{u\in
H^1({\Bbb R}^d)\ \big|\ \|u\|_{X_d}^2=\|u\|_{L^2}^2+\|\nabla
u\|_{L^2}^2+\int_{\Bbb
R^d}V_d(\eta)|u(\eta)|^2\,d\eta<\infty\right\},$$ and
 the unit sphere of $X_d$ defined as
$$
S_d=X_d\bigcap\{u\in L^2(\Bbb R^d)\big|\,\|u\|_{L^2(\Bbb R^d)}=1\}.
$$


This paper is organized as follows. In Sections 2, 3 and 4, we study
quasi-2D equation I (\ref{gpe2d}), II (\ref{gpe2d2}) and quasi-1D
equation (\ref{gpe1d}), respectively. In each section, we first
establish existence and uniqueness as well as nonexistence of the
ground state under different parameter regimes, and then study the
well-posedness of the corresponding  Cauchy problem. In Section 5,
we rigourously prove the validity of dimension reduction from 3D
GPPS (\ref{gpe})-(\ref{poisson}) to 2D and 1D in the weak
interaction regimes. Our approach is based on \textit{a-priori}
estimates from the energy and mass conservation together with the
Strichartz estimates.

Throughout the paper, we adopt the standard notation of Sobolev
space and use $\|f\|_p^p:=\int_{{\Bbb R}^d} |f(\eta)|^p\,d\eta$ for
$p\in(0,\infty)$ when there is no confusion about the space $\Bbb R^d$,
denote $C$ as a generic constant which is
independent of $\vep$, let  $X^\ast$ as the dual space of $X$, and
adopt the Fourier transform of a function
 $f(\eta)\in L^1(\Bbb R^d)$  as
\be \hat{f}(\xi)=\int_{\Bbb R^d}f(\eta)e^{-i\xi\cdot\eta}\,d\eta,\qquad
\xi\in\Bbb R^d. \ee

\section{Results for the quasi-2D equation I}
\setcounter{equation}{0}

In this section, we prove existence and uniqueness as well as
nonexistence of  ground states for the quasi-2D equation I under
different parameter regimes and local (global) existence for the
Cauchy problem. For considering the ground state in 2D, let $C_b$ be
the best constant from the Gagliardo-Nirenberg inequality
\cite{Weinstein}, i.e. \be C_b:=\inf_{0\ne f\in H^1({\Bbb R}^2)} \
\frac{\|\nabla f\|_{L^2({\Bbb R}^2)}^2\cdot\|f\|_{L^2({\Bbb
R}^2)}^2} {\|f\|_{L^4({\Bbb R}^2)}^4}. \ee

\subsection{Existence and uniqueness of ground state}
Associated to the quasi-2D equation I (\ref{gpe2d})-(\ref{u2d1}),
the energy is \be\label{ener2d} E_{2D}(\Phi)=\int_{\Bbb
R^2}\left[\frac12|\nabla\Phi|^2+V_2(\bx)|\Phi|^2+
\frac{\beta-\lambda+3n_3^2\lambda}{2\sqrt{2\pi}\,\vep}|\Phi|^4-\frac{3\lambda}{4}|\Phi|^2
\widetilde{\varphi^{2D}}\right]\,d\bx, \quad \Phi\in X_2,\ee  where
\be\widetilde{\varphi^{2D}}=\left(\p_{\bn_\perp\bn_\perp}-n_3^2\Delta\right)\varphi^{2D},
\qquad \varphi^{2D}=U_\vep^{2D}*|\Phi|^2.\ee

The ground state $\Phi_g\in S_2$ of  (\ref{gpe2d}) is  the minimizer
of the nonconvex minimization problem: \be \mbox{Find } \Phi_g\in
S_2,\quad\mbox{such that }E_{2D}(\Phi_g)=\min\limits_{\Phi\in
S_2}E_{2D}(\Phi). \ee

For the ground state, we have the following results:

\begin{theorem}\label{thm1}(Existence and uniqueness of ground state)
Assume $0\leq V_2(\bx)\in L_{loc}^\infty(\Bbb R^2)$  and
 $\lim\limits_{|\bx|\to\infty}V_2(\bx)=\infty$, then
we have

(i) There exists a ground state $\Phi_g\in S_2$ of the system
(\ref{gpe2d})-(\ref{u2d1}) if one of the following conditions holds
\begin{quote}

(A$1$) $\lambda\ge0$ and $\beta-\lambda> -\sqrt{2\pi} C_b\,\vep$;

(A$2$) $\lambda<0$ and
$\beta+\frac{1}{2}(1+3|2n_3^2-1|)\lambda>-\sqrt{2\pi} C_b\,\vep$.
\end{quote}

(ii) The positive ground state $|\Phi_g|$ is unique under one of the
following conditions:

\begin{quote}
(A1$^\prime$) $\lambda\ge0$ and $\beta-\lambda\ge 0$;

(A2$^\prime$) $\lambda<0$ and
$\beta+\frac12(1+3|2n_3^2-1|)\lambda\ge0$.
\end{quote}
 Moreover, any ground state is of the form $\Phi_g=e^{i\theta_0}|\Phi_g|$
for some constant $\theta_0\in\Bbb R$.

(iii) If $\beta+\frac12\lambda(1-3n_3^2)<-\sqrt{2\pi} C_b\,\vep$,
there exists no ground state of the equation (\ref{gpe2d}).
\end{theorem}

In order to prove this theorem, we first study the property of the
nonlocal term.

\begin{lemma}\label{lem1}(Kernel $U_\vep^{2D}$ in (\ref{u2d1}))
 For any real function $f(\bx)$ in the Schwartz
space ${\mathcal{S}}(\Bbb R^2)$, we have \be\label{u2dcov}
\widehat{U_\vep^{2D}*f}(\xi)=\hat{f}(\xi)\,\widehat{U_\vep^{2D}}(\xi)=\frac{\hat{f}(\xi)}{\pi}\int_{\Bbb
R}\frac{e^{-\vep^2s^2/2}}{|\xi|^2+s^2}ds, \qquad f\in
{\mathcal{S}}(\Bbb R^2). \ee Moreover, define the operator
$$T_{\alpha\alpha^\prime}(f)=\p_{\alpha\alpha^\prime}(U_\vep^{2D}*f),\qquad \alpha,\alpha^\prime=x,y,
$$
then  we have \be \label{bdTjk}\|T_{\alpha\alpha^\prime}f\|_2\leq
\frac{\sqrt{2}}{\sqrt{\pi}\,\vep}\|f\|_2,\qquad \|T_{\alpha\alpha^\prime}f\|_2\leq \|\nabla
f\|_2, \ee hence $T_{\alpha\alpha^\prime}$ can be extended to a bounded linear
operator from $L^2(\Bbb R^2)$ to $L^2(\Bbb R^2)$.
\end{lemma}

\begin{proof} From (\ref{u2d1}), we have  \be
|U_\vep^{2D}(\bx)|=\left|\frac{1}{2\sqrt{2}\pi^{3/2}} \int_{\Bbb
R}\frac{e^{-s^2/2}}{\sqrt{|\bx|^2+ \vep^2 s^2}}\,ds\right|\leq
\frac{1}{2\pi|\bx|},\quad {\bf 0}\ne \bx\in {\Bbb R}^2. \ee  This
immediately implies that $U_\vep^{2D}*g$ is well-defined for any
$g\in L^1(\Bbb R^2)\bigcap L^2(\Bbb R^2)$  since the right hand side
in the above inequality is the singular kernel of Riesz potential.
Re-write $U^{2D}_\vep(\bx)$  as \cite{Bao2}
\[ U^{2D}_\vep(\bx)=\frac{1}{2\pi}\int_{\Bbb
R^2}\frac{w_\vep^2(z)w_\vep^2(z^\prime)}
{\sqrt{|\bx|^2+(z-z^\prime)^2}}dzdz^\prime, \qquad \bx\in{\Bbb R}^2,
\]
 using the Plancherel
formula, we get \be
\widehat{U_\vep^{2D}}(\xi_1,\xi_2)=\frac{1}{\pi}\int_{\Bbb
R}\frac{\widehat{w_\vep^2}(\xi_3)\overline{\widehat{w_\vep^2}}(\xi_3)}
{\xi_1^2+\xi_2^2+\xi_3^2}d\xi_3 =\frac{1}{\pi}\int_{\Bbb
R}\frac{e^{-\vep^2s^2/2}}{|\xi|^2+s^2}ds,\qquad
\xi=(\xi_1,\xi_2)\in{\Bbb R}^2, \ee which immediately implies
(\ref{u2dcov}). Here $\bar{c}$ denotes the complex conjugate of $c$.
Concerning $T_{\alpha\alpha^\prime}$, we only need to prove the
results for  $T_{xx}$ since others are similar.   Applying the Fourier transform, we have  \be
\left|\widehat{T_{xx}f}(\xi)\right|=\left|\frac{\hat{f}(\xi)}{\pi}\int_{\Bbb
R}\frac{e^{-\vep^2s^2/2}\xi_1^2}{|\xi|^2+s^2}ds\right|\leq
\frac{\left|\hat{f}(\xi)\right|}{\pi}\int_{\Bbb
R}e^{-\vep^2s^2/2}ds=\frac{\sqrt{2}}{\sqrt{\pi}\,\vep}\left|\hat{f}(\xi)\right|,\quad
\xi\in{\Bbb R}^2. \ee Thus we can get the first inequality in
(\ref{bdTjk}) and know that $T_{xx}:\ L^2\to L^2$ is bounded.
Moreover, from \be
\left|\widehat{T_{xx}f}(\xi)\right|=\left|\frac{\hat{f}(\xi)}{\pi}\int_{\Bbb
R}\frac{e^{-\vep^2s^2/2}\xi_1^2}{|\xi|^2+s^2}ds\right|\leq
\frac{|\hat{f}(\xi)|\;|\xi_1|^2}{\pi}\int_{\Bbb
R}\frac{1}{|\xi|^2+s^2}ds\leq |\xi|\;|\hat{f}(\xi)|, \ee we  obtain
the second inequality in (\ref{bdTjk}) and know that $T_{xx}:
H^1\to L^2$ is bounded too.  \end{proof}

\begin{remark}\label{LpboundT}In fact, $T_{\alpha\alpha^\prime}$ is bounded from $L^p\to L^p$,
i.e., there exists $C_p>0$ independent of $\vep$, such that \be
\|T_{\alpha\alpha^\prime}(f)\|_{p}\leq \frac{C_p}{\vep}\|f\|_{p},\quad p\in
(1,\infty). \ee This can be obtained by Minkowski inequality and $L^p$ estimates for
Poisson equation.
\end{remark}

\begin{lemma}\label{lem2} For the energy  $E_{2D}(\cdot)$ in (\ref{ener2d}), we have

(i) For any $\Phi\in S_2$, denote $\rho(\bx)=|\Phi(\bx)|^2$, then we have \be E_{2D}(\Phi)\geq
E_{2D}(|\Phi|)=E_{2D}\left(\sqrt{\rho}\right), \qquad \forall \Phi\in S_2,\ee so
the ground state $\Phi_g$ of (\ref{ener2d}) is of the form
$e^{i\theta_0}|\Phi_g|$ for some constant $\theta_0\in \Bbb R$.

(ii) Under the condition (A1) or (A2) in Theorem \ref{thm1},
$E_{2D}(\sqrt{\rho})$ is bounded below.

(iii) Under the  condition (A1$^\prime$) or (A2$^\prime$) in
Theorem \ref{thm1},  $E_{2D}(\sqrt{\rho})$ is strictly convex.
\end{lemma}

\begin{proof} (i) For any $\Phi\in S_2$, then  $|\Phi|\in S_2$, and
a simple calculation shows \be E_{2D}(\Phi)-E_{2D}(|\Phi|)=\frac
12\|\nabla\Phi\|_2^2-\frac12\|\nabla|\Phi|\,\|_2^2\ge0, \qquad
\Phi\in S_2, \ee where the equality holds iff \cite{Lie} \be
|\nabla\Phi(\bx)|=\nabla|\Phi(\bx)|,\quad \hbox{a.e. } \bx\in \Bbb
R^2, \ee
 which is equivalent to
\be \Phi(\bx)=e^{i\theta_0}|\Phi(\bx)|,\quad\hbox{ for some }\
\theta_0\in\Bbb R. \ee Then the conclusion follows.

(ii) For $\sqrt{\rho}=\Phi\in S_2$, we split the energy $E_{2D}$
into two parts, i.e. \be\label{enr2dsplt}
E_{2D}(\Phi)=E_{1}(\Phi)+E_2(\Phi)=E_{1}(\sqrt{\rho})+E_2(\sqrt{\rho}),
\ee where \bea\label{e11} &&E_1(\sqrt{\rho})=\int_{\Bbb
R^2}\left[\frac12|\nabla \sqrt{\rho}|^2+V_2(\bx)\rho\right]d\bx, \\
\label{e12} &&E_2(\sqrt{\rho})=\int_{\Bbb
R^2}\left[\frac{\beta-\lambda+3n_3^2\lambda}{2\sqrt{2\pi}\,\vep}|\rho|^2-\frac{3\lambda}{4}\rho
\widetilde{\varphi^{2D}}\right]\,d\bx, \label{e1e2} \eea with
\be\label{eq:2d:nonlocal}
\widetilde{\varphi^{2D}}=\left(\p_{\bn_\perp\bn_\perp}-n_3^2\Delta\right)U_\vep^{2D}*\rho.
\ee
  Applying the Plancherel formula and Lemma \ref{lem1}, there holds
\bea
\int_{\Bbb R^2}\widetilde{\varphi^{2D}}(\bx)\rho(\bx)\,d\bx&=&\frac{1}{4\pi^2}
\int_{\Bbb R^2}\widehat{\widetilde{\varphi^{2D}}}(\xi)\bar{\hat{\rho}}(\xi)d\xi\nonumber\\
&=&\frac{-1}{4\pi^3}\int_{\Bbb
R^3}\frac{\left((n_1\xi_1+n_2\xi_2)^2-n_3^2
|\xi|^2\right)e^{-\vep^2s^2/2}}{|\xi|^2+s^2}|\hat{\rho}|^2dsd\xi.
\label{e1e22}\eea Recalling the Cauchy inequality and
$n_1^2+n_2^2+n_3^2=1$, we have \be
-n_3^2|\xi|^2\leq(n_1\xi_1+n_2\xi_2)^2-n_3^2|\xi|^2\leq
(1-2n_3^2)|\xi|^2, \qquad \xi\in{\Bbb R}^2. \ee Denoting
$C_0=\max\{|n_3^2|,|1-2n_3^2|\}$, we can derive that
\be\label{eq:u2d:bound} \left|\int_{\Bbb
R^2}\widetilde{\varphi^{2D}}(\bx)\rho(\bx)\,d\bx\right|\leq
\frac{C_0}{4\pi^3}\int_{\Bbb
R^3}e^{-\vep^2s^2/2}|\hat{\rho}|^2dsd\xi
=\frac{\sqrt{2}C_0}{\sqrt{\pi}\,\vep}\|\rho\|_2^2. \ee Hence,
$E_2(\sqrt{\rho})$ can be bounded below by $\|\rho\|_2^2$. In fact,
under the condition (A1), i.e. $\lambda\ge0$ and $\beta-\lambda>
-\sqrt{2\pi} C_b\,\vep$, we have \bea E_{2}(\sqrt{\rho})&\ge&
\frac{\beta-\lambda+3n_3^2\lambda}{2\sqrt{2\pi}\,\vep}\|\rho\|_2^2-
\frac{3\sqrt{2}n_3^2\lambda}{4\sqrt{\pi}\,\vep}\|\rho\|_2^2 >
-\frac{C_b}{2}\|\rho\|_2^2. \eea Similarly, under the condition
(A2), if $\lambda<0$ and $n_3^2\ge\frac12$, then \be
E_{2}(\sqrt{\rho})\ge\frac{\beta-\lambda+3n_3^2\lambda}{2\sqrt{2\pi}\,\vep}
\|\rho\|_2^2>-\frac{C_b}{2}\|\rho\|_2^2; \ee and  if $\lambda<0$
and $n_3^2<\frac12$, then \be
E_{2}(\sqrt{\rho})\ge\frac{\beta-\lambda+3n_3^2\lambda}{2\sqrt{2\pi}\,\vep}\|\rho\|_2^2+
\frac{3\sqrt{2}(1-2n_3^2)\lambda}{4\sqrt{\pi}\,\vep}\|\rho\|_2^2>-\frac{C_b}{2}\|\rho\|_2^2.
\ee Recalling the choice of the best constant $C_b$, under either
condition (A1) or (A2), the energy \be
E_{2D}(\sqrt{\rho})=E_1(\sqrt{\rho})+E_2(\sqrt{\rho})>\frac12
\|\nabla\sqrt{\rho}\|_2^2-\frac{C_b}{2}\|\rho\|_2^2\ge0. \ee

(iii) Again, we split the energy as (\ref{enr2dsplt}). It is well
known that $E_1(\sqrt{\rho})$ is strictly convex in $\rho$
\cite{Lie}. It remains to show that $E_2(\sqrt{\rho})$ is convex in
$\rho$. For any real function $u\in L^1(\Bbb R^2)\cap L^2(\Bbb
R^2)$, let
 \be
 H(u)=\int_{\Bbb R^2}\left[\frac{\beta-\lambda+3n_3^2\lambda}
 {2\sqrt{2\pi}\,\vep}|u|^2-\frac{3\lambda}{4} u
\left(\p_{\bn_\perp\bn_\perp}-n_3^2\Delta_\perp\right)(U_\vep^{2D}*u)\right]\,d\bx.
 \ee
 Then $E_2(\sqrt{\rho})=H(\rho)$. It suffices to show that $H(\rho)$ is convex in $\rho$.
 For this purpose, let $\sqrt{\rho_1}=\Phi_1\in S_2$ and $\sqrt{\rho_2}=\Phi_2\in S_2$,
 for any $\theta\in[0,1]$, consider $\rho_{_{\theta}}=\theta\rho_1+(1-\theta)\rho_2$
 and  $\sqrt{\rho_{_{\theta}}}\in S_2$,
 then we compute directly and get
\bea \theta
H(\rho_1)+(1-\theta)H(\rho_2)-H(\rho_{_\theta})=\theta(1-\theta)H(\rho_1-\rho_2).
\eea Similar as (\ref{e1e22}), looking at the Fourier domain, we can
obtain the lower bounds for $H(\rho_1-\rho_2)$ under the condition
(A1$^\prime$) or (A2$^\prime$), while replacing $C_b$ with $0$ in
the above proof of (ii), i.e., \be H(\rho_1-\rho_2)\ge0. \ee This
shows that $H(\rho)$, i.e.  $E_2(\sqrt{\rho})$, is convex in $\rho$.
Thus  $E_{2D}(\sqrt{\rho})$ is strictly convex in $\rho$. \end{proof}

\bigskip

{{\bf Proof of Theorem \ref{thm1}}:} (i) We first prove the
existence results. Lemma \ref{lem2} ensures that there exists a
minimizing sequence of  nonnegative function
$\{\Phi^n\}_{n=0}^\infty\subset S_2$, such that $
\lim\limits_{n\to\infty}E_{2D}(\Phi^n)=\inf\limits_{\Phi\in
S_2}E(\Phi). $ Then, under condition (A1) or (A2),  there exists a
constant $C$ such that \be
\|\nabla\Phi^n\|_2+\|\Phi^n\|_4+\int_{\Bbb R^2}
V_2(\bx)|\Phi^n(\bx)|^2d\bx\le C, \qquad n\ge0.\ee Therefore
$\Phi^n$ belongs to a weakly compact set in $L^4(\Bbb R^2)$,
$H^1(\Bbb R^2)$, and $L^2_{V_2}(\Bbb R^2)$ with a weighted
$L^2$-norm given by $\|\Phi\|_{L_{V_2}}=[\int_{{\Bbb
R}^2}|\Phi(\bx)|^2V_2(\bx)d\bx]^{1/2}$. Thus, there exists a
$\Phi^\infty\in W:=H^1(\Bbb R^2)\bigcap L^2_{V_2}(\Bbb R^2)\bigcap
L^4(\Bbb R^2)$ and a subsequence of $\{\Phi^n\}_{n=0}^\infty$ (which
we denote as the original sequence for simplicity), such that \be
\label{conveg0}
 \Phi^n\rightharpoonup\Phi^\infty,\quad  \mbox{in } W,\qquad\quad
\nabla \Phi^n\rightharpoonup\nabla\Phi^\infty,\quad \mbox{in } L^2.
\ee The confining condition
$\lim\limits_{|\bx|\to\infty}V_2(\bx)=\infty$ will give that
$\|\Phi^\infty\|_2=1$ \cite{Lieb,Bao1,Bao0}. Hence $\Phi^\infty \in
S_2$ and $\Phi^n\to\Phi^\infty$ in $L^2(\Bbb R^2)$ due to the
$L^2$-norm convergence and weak convergence of
$\{\Phi^n\}_{n=0}^\infty$. By the lower semi-continuity of the
$H^1$- and $L^2_{V_2}$-norm, for $E_1$ in (\ref{e11}), we know \be
E_1(\Phi^\infty)\leq\liminf\limits_{n\to\infty}E_1(\Phi^n). \ee By
the Sobolev inequality, there exists $C(p)>0$ depending on $p\ge2$,
such that $\|\Phi^n\|_p\leq C(p) (\|\nabla
\Phi^n\|_2+\|\Phi^n\|_2)\leq C(p)(1+C)$, uniformly  for $n\ge0$.
Applying the H\"{o}lder's inequality, we have \be
\|(\Phi^n)^2-(\Phi^\infty)^2\|_2^2\leq
C_1(\|\Phi^n\|_6^3+\|\Phi^n\|_6^3)\|\Phi^n-\Phi^\infty\|_2, \ee
which shows $\rho^n=(\Phi^n)^2\to\rho^\infty=(\Phi^\infty)^2\mbox{ in }
L^2(\Bbb R^2)$. Using the Fourier transform of $U_\vep^{2D}$
 in Lemma \ref{lem1} and (\ref{eq:u2d:bound}), it is easy to
 derive the convergence for $E_2$ in (\ref{e12}), i.e.
 \be
 E_2(\Phi^\infty)=\lim\limits_{n\to\infty}E_2(\Phi^n).
 \ee
 Hence
 \be
 E_{2D}(\Phi^\infty)=E_1(\Phi^\infty)+E_2(\Phi^\infty)\leq\liminf\limits_{n\to\infty}E_{2D}(\Phi^n).
 \ee
 Now, we see that $\Phi^\infty$ is indeed a minimizer.
 For the uniqueness part, it is straightforward by  the strict convexity of
 $E_{2D}(\sqrt{\rho})$ shown in Lemma \ref{lem2}.

 (ii) Since the nonlinear term in the equation behaviors as a cubic nonlinearity,
 it is natural to consider the following. Let $\Phi\in S_2$ be a real
 function that attains the best constant $C_b$ \cite{Weinstein},
 then $\Phi(\bx)$ is radially symmetric. Choose $\Phi_{\delta}(\bx)=\delta^{-1}
 \Phi(\delta^{-1}\bx)$ with $\delta>0$,
 then $\Phi_{\delta}\in S_2$. Denote $\varphi_\delta=\left(\p_{\bn_\perp\bn_\perp}
 -n_3^2\Delta_\perp\right)(U_\vep^{2D}*|\Phi_\delta|^2)$,
 by the same computation as in Lemma \ref{lem2}, there holds
 \bea
 \int_{\Bbb R^2}\varphi_\delta |\Phi_\delta|^2\,d\bx&=&\frac{-1}
 {4\pi^3}\int_{\Bbb R^3}\frac{(n_1\xi_1+n_2\xi_2)^2-n_3^2|\xi|^2}{|\xi|^2+s^2}
 e^{-\vep^2s^2/2}\left|\,\widehat{|\Phi|^2}(\delta\xi)\right|^2\,dsd\xi\nonumber
 \\
 &=&\frac{-1}{4\pi^3\delta^2}\int_{\Bbb R^3}\frac{(n_1\xi_1+n_2\xi_2)^2
 -n_3^2|\xi|^2}{|\xi|^2+\delta^2 s^2}e^{-\vep^2s^2/2}\left|\,\widehat{|\Phi|^2}(\xi)
 \right|^2\,dsd\xi.\nonumber
 \eea
 Using the fact that $\Phi(\bx)$ is radially symmetric,  $\widehat{|\Phi|^2}(\xi)$
 is also radially symmetric, then we obtain
 \be
 \int_{\Bbb R^2}\varphi_\delta |\Phi_\delta|^2\,d\bx=-
 \frac{(n_1^2+n_2^2-2n_3^2)+o(1)}{\sqrt{2\pi}\,\vep\delta^2}\|\Phi\|_4^4,\qquad\mbox{as }\ \delta\to 0^+.
 \ee
 Hence, as $\delta\to0^+$, we get
 \be
 E_{2D}(\Phi_\delta)=\frac{1}{2\delta^2}\left(\|\nabla\Phi\|_2^2+\frac{\beta+\frac12\lambda(1-3n_3^2)+o(1)}
 {\sqrt{2\pi}\,\vep}\|\Phi\|_4^4\right)+\int_{\Bbb R^2}V_2(\delta\bx)|\Phi|^2(\bx)d\bx.\nonumber
 \ee
Recalling that $\|\nabla\Phi\|_2^2=C_b\|\Phi\|_4^4$, we know
$\lim\limits_{\delta\to0^+}E_{2D}(\Phi_\delta)=-\infty$ if
$\beta+\frac12\lambda(1-3n_3^2)<-\sqrt{2\pi} C_b\,\vep$, i.e. there is
no ground state in this case. \hfill $\Box$

\subsection{Well-posedness for the Cauchy problem}

Here, we  study the well-posedness of the Cauchy problem
corresponding to the quasi-2D equation I (\ref{gpe2d})-(\ref{u2d1}).
Using the Fourier transform of the kernel $U_\vep^{2D}$ in Lemma
\ref{lem1}, it is straightforward to see that the nonlinear term
introduced by $U_\vep^{2D}$ behaves like cubic term. Thus, those
methods for classic cubic nonlinear Schr\"{o}dinger equation would
apply \cite{Cazen,Weinstein,Sulem}. In particular, we have the
following theorem concerning the Cauchy problem of
(\ref{gpe2d})-(\ref{u2d1}).

\begin{theorem}\label{thm1dy}
(Well-posedness of Cauchy problem) Suppose the real-valued trap
potential satisfies $V_2(\bx)\ge0$ for $\bx\in{\Bbb R}^2$ and
\be\label{cond:v2} V_2(\bx)\in C^\infty(\Bbb R^2) \hbox{ and
}D^{{\bf k}} V_2(\bx)\in L^\infty(\Bbb R^2),\qquad \hbox{for all }
{\bf k}\in{\Bbb N}_0^2\  \hbox{with}\  |{\bf k}|\ge 2,\ee then we have

(i) For any initial data $\phi(\bx,t=0)=\phi_0(\bx)\in X_2$,
 there exists a
$T_{\mbox{\rm max}}\in(0,+\infty]$ such that the problem
(\ref{gpe2d})-(\ref{u2d1})
 has a unique maximal solution
$\phi\in C\left([0,T_{\mbox{\rm max}}),X_2\right)$. It is maximal in
the sense that if $T_{\mbox{\rm max}}<\infty$, then
$\|\phi(\cdot,t)\|_{X_2}\to\infty$ when  $t\to T^-_{\mbox{\rm
max}}$.

(ii) As long as the solution $\phi(\bx,t)$ remains in the energy
space $X_2$, the {\sl $L^2$-norm} $\|\phi(\cdot,t)\|_2$ and {\sl
energy} $E_{2D}(\phi(\cdot,t))$ in (\ref{ener2d}) are conserved for
$t\in[0,T_{\rm max})$.

(iii) Under either condition (A1) or (A2) in  Theorem \ref{thm1}
with constant $C_b$ being replaced by $C_b/\|\phi_0\|_2^2$,
the solution of (\ref{gpe2d})-(\ref{u2d1}) is global in time, i.e.,
  $T_{\mbox{max}}=\infty$.
\end{theorem}

\begin{proof} The proof is standard. We shall use the known results
for semi-linear Schr\"{o}dinger equation \cite{Cazen}.  For $\phi\in
X_2$, denote $\rho=|\phi|^2$ and consider the following \bea
&&G(\phi,\bar{\phi}):=G(\rho)=\frac{1}{2}\int_{\Bbb R^2}|\phi|^2
\left(\p_{\bn_\perp\bn_\perp}-n_3^2\Delta\right)(U^{2D}_\vep*|\phi|^2)\,d\bx,\nonumber\\
&&g(\phi)=\frac{\delta G(\phi,\bar{\phi})}{\delta
\bar{\phi}}=\phi\left(\p_{\bn_\perp\bn_\perp}-n_3^2\Delta\right)(U^{2D}_\vep*|\phi|^2),
\qquad \phi\in X_2.\nonumber \eea Then the equations
(\ref{gpe2d})-(\ref{u2d1}) read \be
i\p_t\phi=\left[-\frac{1}{2}\Delta+V_2(\bx)\right]\phi+\beta_0|\phi|^2\phi-3\lambda
g(\phi), \qquad \bx\in{\Bbb R}^2, \quad t>0,\ee where
$\beta_0=\frac{\beta-\lambda+3n_3^2\lambda}{\sqrt{2\pi}\,\vep}$.
Using the $L^p$ boundedness of $T_{jk}$ (cf. Lemma \ref{lem1} and
Remark \ref{LpboundT}) and the Sobolev inequality,  for
$\|u\|_{X_2}+\|v\|_{X_2}\leq M$, it is easy to prove the following
\be \|g(u)-g(v)\|_{4/3}\leq C(M) \|u-v\|_4. \ee In view of the
standard Theorems 9.2.1, 4.12.1 and 5.7.1 in \cite{Cazen} and
\cite{Sulem} for the well-posedness of the nonlinear Schr\"{o}dinger
equation, we can obtain the results (i) and (ii) immediately. The
global existence (iii) comes from the uniform bound for
$\|\phi(\cdot,t)\|_{X_2}$ which can be derived from energy and
$L^2$-norm conservation. \end{proof}

When the initial data is small, there also exists global solutions
\cite{Cazen,Carles}. Otherwise, blow-up may happen in finite time,
and we have the following results.

\begin{theorem}\label{blowup-1}(Finite time blow-up) For any initial data
$\phi(\bx,t=0)=\phi_0(\bx)\in X_2$ with $\int_{\Bbb
R^2}|\bx|^2|\phi_0(\bx)|^2\,d\bx<\infty$,  if conditions (A1) and (A2)
with constant $C_b$ being replaced by $C_b/\|\phi_0\|_2^2$
  are not satisfied and assume $V_2(\bx)$ satisfies $2V_2(\bx)+ \bx\cdot
 \nabla V_2(\bx)\ge0$,  and let $\phi:=\phi(\bx,t)$ be
the solution of the problem (\ref{gpe2d}), there exists finite time
blow-up, i.e., $T_{\hbox{max}}<\infty$,  if $\lambda=0$, or
$\lambda>0$ and $n_3^2\ge\frac12$, and one of the following holds:

(i) $E_{2D}(\phi_0)<0$;

(ii) $E_{2D}(\phi_0)=0$ and ${\rm Im}\left(\int_{\Bbb
R^2}\bar{\phi}_0(\bx)\ (\bx\cdot\nabla\phi_0(\bx))\,d\bx\right)<0$;

(iii) $E_{2D}(\phi_0)>0$ and ${\rm Im}\left(\int_{\Bbb R^2}
\bar{\phi}_0(\bx)\ (\bx\cdot\nabla\phi_0(\bx))\,d\bx\right)
<-\sqrt{2E_{2D}(\phi_0)}\|\bx\phi_0\|_{2}$;

\noindent where ${\rm Im}(f)$ denotes the imaginary part of $f$.
\end{theorem}
\begin{proof}
Define the variance \be\label{dtv001}
\sg_V(t):=\sg_V(\phi(\cdot,t))=\int_{\Bbb
R^2}|\bx|^2|\phi(\bx,t)|^2\,d\bx=\sg_{x}(t)+\sg_{y}(t),\qquad
t\ge0,\ee where \be
\label{dtap01}\sg_\alpha(t):=\sg_\alpha(\phi(\cdot,t))=\int_{\Bbb
R^2}\alpha^2|\phi(\bx,t)|^2\,d\bx, \qquad \alpha=x,\ y.\ee

For $\alpha=x$, or $y$, differentiating (\ref{dtap01}) with respect
to $t$, integrating by parts, we get \be
\frac{d}{dt}\sg_\alpha(t)=-i\int_{\Bbb
R^2}\left[\alpha\bar{\phi}(\bx,t)\p_{\alpha}\phi(\bx,t)-
\alpha\phi(\bx,t)\p_{\alpha}\bar{\phi}(\bx,t)\right]\,d\bx, \qquad
t\ge0.\ee Similarly, we have \be\label{d2ap22}
\frac{d^2}{dt^2}\sg_\alpha(t)=\int_{\Bbb
R^2}\left[2|\p_{\alpha}\phi|^2+\beta_0
|\phi|^4+3\lambda|\phi|^2\alpha\p_{\alpha}(\p_{\bn_\perp\bn_\perp}-n_3^2\Delta)
\varphi-2\alpha|\phi|^2\p_{\alpha}V_2(\bx)\right]\,d\bx,\ee where
$\beta_0=\frac{\beta-\lambda+3\lambda n_3^2}{\sqrt{2\pi}\,\vep}$,
$\varphi=U^{2D}_\vep *|\phi|^2$. Writing $\rho=|\phi|^2$,
$\tilde{\varphi}=(\p_{\bn_\perp\bn_\perp}-n_3^2\Delta) \varphi$,
$n_\xi(\xi)=(n_1\xi_1+n_2\xi_2)^2-n_3^2|\xi|^2$ and noticing that $\rho$
is a real function,  by the Plancherel formula,  we have \beas
\int_{\Bbb R^2}|\phi|^2\left(\bx\cdot\nabla
\tilde{\varphi}\right)\,d\bx&=&\frac{-1}{4\pi^2}\int_{\Bbb R^2}\hat{\rho}(\xi)\,
\nabla\cdot\left(\xi\overline{\hat{\tilde{\varphi}}}\right)\,d\xi
=\frac{1}{4\pi^2}\int_{\Bbb R^2}\hat{\rho}(\xi)\,
\nabla\cdot\left(\xi n_\xi\widehat{U^{2D}_\vep}\overline{\hat{\rho}}\right)\,d\xi\\
&=&\frac{1}{4\pi^2}\int_{\Bbb R^2}\hat{\rho}\,\left(
\overline{\hat{\rho}}\nabla(\xi n_\xi\widehat{U^{2D}_\vep})+n_\xi\widehat{U^{2D}_\vep}
\xi\cdot \nabla\overline{\hat{\rho}} \right)\,d\xi\\
&=&\frac{1}{4\pi^2}\int_{\Bbb R^2}\left( |\hat{\rho}|^2\nabla(\xi n_\xi\widehat{U^{2D}_\vep})
+n_\xi\widehat{U^{2D}_\vep}\xi\cdot \frac12\nabla|\hat{\rho}|^2 \right)\,d\xi\\
&=&\frac{1}{4\pi^2}\int_{\Bbb R^2} (n_\xi\widehat{U_\vep^{2D}}+\frac12\xi\cdot
\nabla(n_\xi\widehat{U^{2D}_\vep})\,)|\hat{\rho}|^2\,d\xi\\
&=&-\int_{\Bbb R^2}|\phi|^2
\tilde{\varphi}\,d\bx+\frac{1}{4\pi^3}\int_{\Bbb R^3}\frac{n_\xi s^2
e^{-\vep^2s^2/2}|\hat{\rho}|^2}{(|\xi|^2+s^2)^2}\,dsd\xi. \eeas
Denote \be I(t):=I(\phi(\cdot,t))=\frac{1}{4\pi^3}\int_{\Bbb
R^3}\frac{n_ \xi s^2
e^{-\vep^2s^2/2}|\hat{\rho}|^2}{(|\xi|^2+s^2)^2}\,dsd\xi,\qquad
t\ge0,\ee using $n_\xi\in[-n_3^2|\xi|^2,(1-2n_3^2)|\xi|^2]$, we
obtain \be
\frac{-\sqrt{2}n_3^2}{\sqrt{\pi}\,\vep}\|\phi(t)\|_4^4\leq I(t)\leq
\frac{\sqrt{2}(1-2n_3^2)}{\sqrt{\pi}\,\vep}\|\phi(t)\|_4^4, \qquad
t\ge0. \ee
 If $\lambda=0$, or $\lambda>0$ and $n_3\ge\frac12$, noticing $\lambda I(t)\leq0$
 in these cases, summing (\ref{d2ap22}) for $\alpha=x$, $y$, and using the energy conservation, we have
\beas
\frac{d^2}{dt^2}\sg_V(t)&=&2\int_{\Bbb
R^2}\left[|\nabla\phi|^2+\beta_0
|\phi|^4+\frac32\lambda|\phi|^2\left(\bx\cdot\nabla
\tilde{\varphi}\right)-|\phi|^2\bx\cdot\nabla V_2(\bx)\right]\,d\bx\\
&=&4E_{2D}(\phi(\cdot,t))+3\lambda I(t)-2\int_{\Bbb R^2}|\phi|^2(2V_2(\bx)+\bx\cdot\nabla V_2(\bx))\,d\bx\\
&\leq&4E_{2D}(\phi(\cdot,t))\equiv 4E_{2D}(\phi_0), \qquad t\ge0.
\eeas Thus,
$$
\sg_V(t)\leq 2E_{2D}(\phi_0)t^2+\sg_V^\prime(0)t+\sg_V(0), \qquad t\ge0,
$$
and the conclusion follows in the same manner as those in
\cite{Sulem,Cazen} for the standard nonlinear Schr\"{o}dinger
equation.
\end{proof}

\section{Results for the quasi-2D equation II}
\setcounter{equation}{0}

In this section, we investigate the existence, uniqueness as well
 as nonexistence of ground state of the quasi-2D equation II (\ref{gpe2d2})
 and the well-posedness of the corresponding Cauchy problem.

\subsection{Existence and uniqueness of ground state}

Associated to the quasi-2D equation II (\ref{gpe2d2}), the energy is
\be\label{ener2d2} \tilde{E}_{2D}(\Phi)=\int_{\Bbb
R^2}\left[\frac12|\nabla\Phi|^2+
V_2(\bx)|\Phi|^2+\frac{\beta-\lambda+3n_3^2\lambda}{2\sqrt{2\pi}\,\vep}|\Phi|^4
-\frac{3\lambda}{4}|\Phi|^2 \varphi\right]\,d\bx, \qquad \Phi\in
X_2,\ee
 where
\be\label{vphu2d22}
\varphi(\bx)=\left(\p_{\bn_\perp\bn_\perp}-n_3^2\Delta\right)((-\Delta)^{-1/2}|\Phi|^2).
\ee The ground state $\Phi_g\in S_2$ of the equation (\ref{gpe2d2})
is defined as the minimizer of the nonconvex minimization problem:
\be \mbox{Find } \Phi_g\in S_2,\quad\mbox{such that
}\tilde{E}_{2D}(\Phi_g)=\min\limits_{\Phi\in
S_2}\tilde{E}_{2D}(\Phi). \ee

For the above ground state, we have  the following results.

\begin{theorem}\label{thm1'}(Existence and uniqueness of
ground state) Assume $0\leq V_2(\bx)\in L_{loc}^\infty(\Bbb R^2)$
 and
$\lim\limits_{|\bx|\to\infty}V_2(\bx)=\infty$, then we have

(i) There exists a ground state $\Phi_g\in S_2$ of the equation
(\ref{gpe2d2}) if one of the following conditions holds
\begin{quote}
(B1) $\lambda=0$ and $\beta> -\sqrt{2\pi}C_b\,\vep$;\\
(B2)  $\lambda>0$, $n_3=0$ and $\beta-\lambda> - \sqrt{2\pi}C_b\,\vep$;\\
(B3)  $\lambda<0$, $n_3^2\ge\frac12$ and
$\beta-(1-3n_3^2)\lambda>-\sqrt{2\pi} C_b\,\vep$.
\end{quote}

(ii)  The positive ground state $|\Phi_g|$ is unique under one of
the following conditions
\begin{quote}
(B1$^\prime$) $\lambda=0$ and $\beta\ge 0$;\\
(B2$^\prime$) $\lambda>0$, $n_3=0$ and $\beta\ge\lambda$;\\
(B3$^\prime$) $\lambda<0$, $n_3^2\ge\frac12$ and
$\beta-(1-3n_3^2)\lambda\ge0$.
\end{quote}
Moreover, any ground state $\Phi_g=e^{i\theta_0}|\Phi_g|$ for some
constant $\theta_0\in\Bbb R$.

(iii) There exists no ground state of the equation (\ref{gpe2d2}) if
one of the following conditions holds
\begin{quote}
(B1$^{\prime\prime}$) $\lambda>0$ and $n_3\neq0$;\\
(B2$^{\prime\prime}$) $\lambda<0$ and  $n_3^2<\frac12$;\\
(B3$^{\prime\prime}$) $\lambda=0$ and $\beta<- \sqrt{2\pi}C_b\,\vep$.
\end{quote}
\end{theorem}

Again, in order to prove this theorem, we first analyze the nonlocal
part in the equation (\ref{gpe2d2}). In fact, following the standard
proof in \cite{EMStein}, we can get

\begin{lemma}\label{lem3}(Property of fractional Poisson
equation (\ref{u2d2})) Assume $f(\bx)$ is a real valued function
good enough, for the fractional Poisson equation
$$
(-\Delta)^{-1/2}\varphi(\bx)=f(\bx),\quad \bx\in\Bbb R^2, \qquad
\lim\limits_{|\bx|\to\infty}\varphi(\bx)=0,
$$
we have
$$
\varphi(\bx)=\int_{\Bbb
R^2}\frac{f(\bx^\prime)}{2\pi|\bx-\bx^\prime|}d\bx^\prime=\left(\frac{1}{2\pi|\bx|}\right)*f,
\qquad \bx\in {\Bbb R}^2,
$$
and the Hardy-Littlewood-Sobolev inequality implies
 \be
 \|\varphi\|_{p^\ast}\leq C_p\|f\|_p,\qquad p^\ast=\frac{2p}{2-p},\qquad p\in(1,2).
 \ee
Moreover,  the first order derivatives of $\varphi$ are the Riesz
transforms of $f$ and satisfy \be \|\p_{\alpha}\varphi\|_q\leq
C_q\|f\|_q,\qquad q\in(1,\infty), \quad \alpha=x,\ y, \ee and the second
order derivatives satisfy \be
\|\p_{\alpha \alpha^\prime}\varphi\|_q=\|\p_{\alpha}\left((-\Delta)^{-1/2}\p_{\alpha^\prime}f\right)\|_q\leq
C_q\|\p_{\alpha^\prime}f\|_q, \qquad q\in(1,\infty), \quad \alpha,\alpha^\prime=x,y. \ee
\end{lemma}

 \begin{remark}\label{rmk} Similar results hold for $T_{\alpha\alpha^\prime}$ defined in Lemma
 \ref{lem1}, i.e.
\be \|T_{\alpha\alpha^\prime}f\|_p\leq C_p \|\nabla f\|_p,\qquad \hbox{for}\  p\in
(1,\infty). \ee
\end{remark}

Since the fractional Poisson operator $(-\Delta)^{-1/2}$ is taken as
an approximation of $U^{2D}_\vep$ (\ref{u2d1}), we consider the
convergence regarding with the derivatives.

\begin{lemma}\label{converge:2D} For any real-valued function $f\in L^p(\Bbb R^2)$, let
\be T^\vep_\alpha(f)=\p_{\alpha}(U^{2D}_\vep*f),\qquad
R_\alpha(f)=\p_{\alpha}(-\Delta)^{-1/2}f, \qquad \alpha=x,y, \ee then  $T^\vep_\alpha$
is bounded from $L^p$ to $L^p$ for $1<p<\infty$ with the bounds
independent of $\vep$. Specially, for any fixed $f\in L^p(\Bbb R^2)$
with $p\in(1,\infty)$, we have \be \lim\limits_{\vep\to
0^+}\|T^\vep_\alpha(f)-R_\alpha(f)\|_p=0,\qquad p\in(1,\infty). \ee
\end{lemma}

\begin{proof} We can write $R_\alpha$ and $T^\vep_\alpha$ as \be
R_\alpha(f)=K_\alpha*f,\qquad T_\alpha^\vep(f)=K_\alpha^\vep*f, \ee where $R_\alpha$ is the
Riesz transform and \be K_\alpha(\bx)=\frac{\alpha}{2\pi|\bx|^3},\quad
K_\alpha^\vep(\bx)=\frac{1}{2\sqrt{2}\pi^{3/2}} \int_{\Bbb
R}\frac{\alpha e^{-s^2/2}}{(|\bx|^2+ \vep^2s^2)^{3/2}}\,ds,\quad
\bx\in{\Bbb R}^2, \quad \alpha=x,y. \ee It is easy to check that
$K_\alpha^\vep$  satisfies
\begin{eqnarray*}
&&|K_\alpha^\vep(\bx)|\leq B|\bx|^{-2}, \qquad |\nabla K_\alpha^\vep(\bx)|\leq B|\bx|^{-3},\qquad |\bx|>0,\\
&&\int_{R_1<|\bx|<R_2}K_j^\vep(\bx)d\bx=0,\qquad 0<R_1<R_2<\infty,
\end{eqnarray*}
for some  $\vep$-independent constant $B$. Then the standard theorem
on singular integrals \cite{EMStein} implies that $T_\alpha^\vep$ is well
defined for $L^p$  functions  and  is bounded from $L^p$ to $L^p$
with $\vep$-independent bound.

Thus, we only need to prove the convergence in $L^2$, other cases
can be derived by an approximation argument and interpolation. For
the  $L^2$ convergence, looking at the Fourier domain, we find that
\beas
\|T_\alpha^\vep(f)-R_\alpha(f)\|_2^2&=&\frac{1}{4\pi^2}\int_{\Bbb R^2}\left|
\hat{f}(\xi)\right|^2\left[\frac{\alpha}{|\xi|}-\frac{\alpha}{\pi}
\int_{\Bbb R}\frac{e^{-\vep^2s^2/2}}{|\xi|^2+s^2}ds\right]^2d\xi\\
&\leq&\frac{1}{4\pi^2}\int_{\Bbb
R^2}\left|\hat{f}(\xi)\right|^2\left[\frac{1}{\pi} \int_{\Bbb
R}\frac{(1-e^{-\vep^2s^2/2)}|\xi|}{|\xi|^2+s^2}ds\right]^2d\xi.
\eeas Notice that for fixed $0\neq \xi\in {\Bbb R}^2$, the dominated convergence
theorem suggests that \be
\lim\limits_{\vep\to0^+}\frac{1}{\pi}\left|\int_{\Bbb
R}\frac{(1-e^{-\vep^2s^2/2)}|\xi|}{|\xi|^2+s^2}ds\right|=0,\ee
 hence, the conclusion in $L^2$ case is obvious by using the
 dominated convergence theorem  again. Using approximation and noticing that
 $L^2\cap L^q$ is dense in $L^p$  for $q\in(1,\infty)$, noticing  the
 uniform bound on $T^\vep_\alpha:L^p\to L^p$ for $p\in(1,\infty)$,
 we can complete the proof. \end{proof}

\begin{lemma}\label{lem4} For the energy  $\tilde{E}_{2D}(\cdot)$ in (\ref{ener2d2}),
the following properties hold

(i) For any $\Phi\in S_2$, denote $\rho(\bx)=|\Phi(\bx)|^2$, then we have \be \tilde{E}_{2D}(\Phi)\geq
\tilde{E}_{2D}(|\Phi|)=\tilde{E}_{2D}\left(\sqrt{\rho}\right), \qquad \forall \Phi\in S_2,\ee so
the ground state $\Phi_g$ of (\ref{ener2d2}) is of the form
$e^{i\theta_0}|\Phi_g|$ for some constant $\theta_0\in \Bbb R$.

(ii) If condition (B1) or (B2) or (B3) in  Theorem \ref{thm1'}
holds, then $\tilde{E}_{2D}$ is bounded below.

(iii) If condition (B1$^\prime$) or (B2$^\prime$) or (B3$^\prime$)
in Theorem \ref{thm1'} holds,  then $\tilde{E}_{2D}(\sqrt{\rho})$ is
strictly convex.
\end{lemma}

\begin{proof} (i) It is similar to the case of Lemma \ref{lem2}.

(ii) Similar as  Lemma \ref{lem2}, for $\Phi\in S_2$, denote
$\rho=|\Phi|^2$,  we only need to consider the lower bound of the following functional
\be\label{Hrho} \tilde{H}(\rho)=-\lambda\int_{\Bbb
R^2}\rho\left(\p_{\bn_\perp\bn_\perp}-n_3^2\Delta\right)[(-\Delta)^{-1/2}\rho]\,d\bx.
\ee Using the Plancherel formula and Cauchy inequality, for
$\lambda<0$ and $n_3^2\ge\frac12$, we have \bea
\tilde{H}(\rho)=\frac{\lambda}{4\pi^2}\int_{\Bbb R^2}\frac{(n_1\xi_1+n_2\xi_2)^2-
n_3^2|\xi_3|^2}{|\xi|}\left|\hat{\rho}(\xi)\right|^2\,d\xi
\ge\frac{\lambda}{4\pi^2}\int_{\Bbb
R^2}(1-2n_3^2)|\xi|\;\left|\hat{\rho}(\xi)\right|^2\,d\xi\ge0. \quad \eea
For $\lambda>0$ and $n_3=0$, it is easy to see
$\tilde{H}(\rho)\ge0$. Hence, assertion (ii) is proven.

(iii)  Similar as  Lemma \ref{lem2}, it is sufficient to prove the
convexity of $\tilde{H}(\rho)$ in $\rho$. For $\sqrt{\rho_1}\in
S_2$, $\sqrt{\rho_2}\in S_2$ and any $\theta\in [0,1]$, denote
$\rho_{_\theta}=\theta\rho_1+(1- \theta)\rho_2$, we have \be \theta
\tilde{H}(\rho_1)+(1-\theta)\tilde{H}(\rho_2)-\tilde{H}(\rho_{_\theta})
=\theta(1-\theta)\tilde{H}(\rho_1-\rho_2),
\ee where the RHS is nonnegative under the given condition, i.e.,
$\tilde{H}(\rho)$ is convex. \end{proof}

\bigskip

{\bf{Proof of Theorem \ref{thm1'}:}} (i) We only need to consider
the existence since the uniqueness is a consequence of the convexity
of $\tilde{E}_{2D}(\sqrt{\rho})$ in Lemma \ref{lem4}. For existence,
we may apply the same arguments in Theorem \ref{thm1}, where
instead,  for sequence $\rho^n=(\Phi^n)^2$, we have to show
\be\label{2dgs} \liminf\limits_{n\to\infty}\tilde{H}(\rho^n)\ge
\tilde{H}(\rho^\infty),\quad \mbox{ with }\
\rho^\infty=|\Phi^\infty|^2. \ee Denote
$$
\varphi^n=\left(\p_{\bn_\perp\bn_\perp}-n_3^2\Delta\right)[(-\Delta)^{-1/2}\rho^n],
\qquad  n=0,1,\ldots,\mbox { or } \ n=\infty.
$$
Using $\Phi^n\to\Phi^\infty$ in $L^2(\Bbb R^2)$ and
$\Phi^n\rightharpoonup \Phi^\infty$ in $H^1(\Bbb R^2)$, then
$\rho^n\to\rho^\infty$ in $L^p(\Bbb R^2)$ $p>1$, and Lemma
\ref{lem3} shows that $\varphi^n\to\varphi^\infty$ in $W^{-1,p^\prime}(\Bbb
R^2)$ which is the dual space of $W^{1,p}$ with
$p^\prime=p/(p-1)$. Thus (\ref{2dgs}) is true and the existence of
ground state follows.

(ii) To prove the nonexistence results, we try to find the case
where $\tilde{E}_{2D}$ doesn't have  lower bound. For any $\Phi\in
S_2$, denote $\rho(\bx)=|\Phi(\bx)|^2$ and let $\theta\in \Bbb R$
such that
$(\cos\theta,\sin\theta)=\frac{1}{\sqrt{n_1^2+n_2^2}}(n_1,n_2)$ when
$n_1^2+n_2^2\neq0$ and $\theta=0$ if $n_1=n_2=0$. For any
$\vep_1,\vep_2>0$, consider the following function \be
\Phi_{\vep_1,\vep_2}(x,y)=\vep_1^{-1/2}\vep_2^{-1/2}\Phi(\vep_1^{-1}
(x\cos\theta+y\sin\theta),\vep_2^{-1}(-x
\sin\theta+y\cos\theta)), \ee denoting
$\rho_{\vep_1,\vep_2}=|\Phi_{\vep_1,\vep_2}|^2$, then \be
\widehat{\rho_{\vep_1,\vep_2}}(\xi_1,\xi_2)=\hat{\rho}
(\vep_1(\xi_1\cos\theta+\xi_2\sin\theta),\vep_2(-\xi_1
\sin\theta+\xi_2\cos\theta)). \ee By the Plancherel formula and
changing variables, we get \bea
\tilde{H}(\rho_{\vep_1,\vep_2})&=&\frac{\lambda}{4\pi^2}\int_{\Bbb R^2}
\frac{(n_1\xi_1+n_2\xi_2)^2-n_3^2|\xi^2|}{|\xi|}|\widehat{\rho_{\vep_1,
\vep_2}}|^2d\xi\nonumber\\
&=&\frac{\lambda}{4\pi^2}\int_{\Bbb R^2}\frac{(n_1^2+n_2^2)\eta_1^2-
n_3^2|\eta|^2}{|\eta|}|\widehat{\rho}|^2(\vep_1\eta_1,\vep_2\eta_2)d\eta\nonumber\\
&=&\frac{\lambda}{4\vep_1^2\vep_2\pi^2}\int_{\Bbb
R^2}\frac{(1-n_3^2)
\eta_1^2-n_3^2(\eta_1^2+\frac{\vep_2^2}{\vep_1^2}\eta_2^2)}{\sqrt{\eta_1^2+
\frac{\vep_2^2}{\vep_1^2}\eta_2^2}}
|\widehat{\rho}|^2(\eta_1,\eta_2)d\eta\nn. \eea Let
$\kappa=\frac{\vep_2}{\vep_1}$, then the  dominated convergence
theorem implies \bea \tilde{H}( \rho_{\vep_1,\vep_2})=\left\{\ba{ll}
\frac{1-2n_3^2+o(1)}{4\vep_1^2\vep_2}\lambda\int_{\Bbb
R^2}|\eta_1|\;\left|\widehat{\rho}(\eta_1,\eta_2)\right|^2\,d\eta,
&\kappa\to0^+, \\
\frac{-n_3^2+o(1)}{4\vep_1^2\vep_2}\lambda\int_{\Bbb
R^2}|\eta_1|\;\left|\widehat{\rho}(\eta_1,\eta_2)\right|^2\,d\eta,
&\kappa\to+\infty. \label{b2256}\ea\right.\qquad \eea For fixed
$\kappa>0$ and letting $\vep_1\to0^+$, we have $\int_{\Bbb
R^2}V_2(\bx)|\Phi_{\vep_1,\vep_2}|^2\,d\bx=O(1)$ and \be
\|\nabla\Phi_{\vep_1,\vep_2}\|_2^2=\frac{1}{\vep_1^2}\|\p_{x}\Phi\|_2^2
+\frac{1}{\vep_2^2}\|\p_{y}\Phi\|_2^2,\qquad
\|\Phi_{\vep_1,\vep_2}\|_4^4=\frac{1}{\vep_1\vep_2}\|\Phi\|_4^4. \ee
Thus under the condition (B1$^{\prime\prime}$), i.e $n_3\neq0$ and
$\lambda>0$, choosing $\kappa$ large enough, we get \be
\tilde{E}_{2D}(\Phi_{\vep_1,\vep_2})=\frac{C_1}{\vep^2_1}+\frac{C_2}{\kappa^2\vep_1^2}+
\frac{C_3}{\kappa\vep_1^2}+C_4\lambda\frac{-n_3^2+o(1)}{\kappa\vep_1^3}+O(1),
\ee where $C_k$ ($k=1,2,3,4$) are  constants independent of
$\kappa$, $\vep_1$ and $C_4>0$. Since $n_3\neq0$, the last term is
negative for $\kappa$ large, sending $\vep_1\to0^+$, one immediately
finds that
$\lim\limits_{\vep_1\to0^+,\vep_2=\kappa\vep_1}\tilde{E}_{2D}(\Phi_{\vep_1,\vep_2})=-\infty$,
which justifies the nonexistence.  Under the condition
(B2$^{\prime\prime}$), i.e. $n_3^2< \frac12$ and $\lambda<0$, by
choosing $\kappa$ small enough in (\ref{b2256}), sending $\vep_1$ to
$0^+$, we will have the same results. Case (B3$^{\prime\prime}$)
will reduce to Theorem \ref{thm1}. \hfill $\Box$

\subsection{Existence results for the Cauchy problem}

Let us consider the Cauchy problem of equation (\ref{gpe2d2}),
noticing the nonlinearity
$\phi(\p_{\bn_\perp\bn_\perp}-n_3^2\Delta)((-\Delta)^{-1/2}|\phi|^2)$
is actually a derivative nonlinearity,  it would bring
significant difficulty in analyzing the dynamic behavior. The common
approach to solve the Schr\"{o}dinger equation is trying to solve
the corresponding integral equation by  fixed point theorem.
However, the loss of order 1 derivative due to the nonlocal term
will cause trouble. This can be overcome by the smoothing effect of
inhomogeneous problem $iu_t+\Delta u=g(x,t)$, which provides a gain
of order 1 derivative \cite{Bour,Kpv0}. When $V_2(\bx)=0$, i.e. 
without external trapping potential which corresponds to the free
expansion of a dipolar BEC after turning off the confinement, the above
approach can be extended straightforward. However, when $V_2(\bx)\neq0$,
i.e. with an external trapping potential, especially a confinement 
trapping potential with    
$\lim_{|\bx|\to\infty}V_2(\bx)=\infty$, 
the approach in \cite{Bour,Kpv0} has some difficulties.  By
configuring that
$(\p_{\bn_\perp\bn_\perp}-n_3^2\Delta)((-\Delta)^{-1/2}|\phi|^2)$ is
almost a first order derivative, we are able to establish the
well-posedness of (\ref{gpe2d2}) with a general external potential 
$V_2(\bx)$ by a different approach.

The Cauchy problem of the Schr\"{o}dinger equation with derivative
nonlinearity has been investigated extensively  in the literatures
\cite{Ho,Kpv1}. Here, we present an existence results in the energy
space with the special structure of our nonlinearity, which will
show that the approximation (\ref{gpe2d2}) of (\ref{gpe2d}) is
reasonable in suitable sense. We are interested in the case of $\lambda\neq 0$.

\begin{theorem}\label{dyn2d2}
(Existence  for Cauchy problem) Suppose the real potential
$V_2(\bx)$ satisfies (\ref{cond:v2}) and
$\lim_{|\bx|\to\infty}V_2(\bx)=\infty$,
and initial value $\phi_0(\bx)\in X_2$,
either condition  (B2) or (B3)  in Theorem \ref{thm1'} holds with
constant $C_b$ being replaced by $C_b/\|\phi_0\|_2^2$, then
there exists a solution $\phi\in L^\infty([0,\infty);X_2)\cap
W^{1,\infty}([0,\infty);X_2^\ast)$ for the Cauchy problem of
(\ref{gpe2d2}). Moreover, there holds for $L^2$ norm and energy
$\tilde{E}_{2D}$ (\ref{ener2d2}) conservation, i.e. \be
\|\phi(\cdot,t)\|_{L^2(\Bbb R^2)}=\|\phi_0\|_{L^2(\Bbb R^2)},\quad
\tilde{E}_{2D} (\phi(\cdot,t))\leq \tilde{E}_{2D}(\phi_0),
\quad\forall t\ge0. \ee
\end{theorem}

\begin{proof} We first consider the Cauchy problem for the following
equation, \be\label{eq:appro} i \p_t
\phi^\delta(\bx,t)=\H_{\bx}^V\phi^\delta+g_1(\phi^\delta)+g_2(\phi^\delta),
\qquad \bx\in{\Bbb R}^2, \quad t>0, \ee with the initial data
$\phi^\delta(\bx,0)=\phi_0(\bx)$,
$\beta_0=\frac{\beta-\lambda+3\lambda n_3^2}{\sqrt{2\pi}\,\vep}$,
$\varphi^\delta=U^{2D}_\delta*|\phi^\delta|^2$ where $U^{2D}_\delta$ is given in
(\ref{u2d1}) as $U^{2D}_\delta(\bx)=\frac{1}{2\sqrt{2}\pi^{3/2}}
 \int_{\Bbb R}\frac{e^{-s^2/2}}{\sqrt{|\bx|^2+
\delta^2s^2}}\,ds$ ($\delta>0$), and \be
 \H_{\bx}^V=-\frac{1}{2}\Delta+V_2(\bx),\quad g_1(\phi^\delta)=
 \beta_0|\phi^\delta|^2\phi^\delta,\quad g_2(\phi^\delta)=
 -\frac{3\lambda}{2}\phi^\delta( \p_{\bn_\perp\bn_\perp}
-n_3^2\Delta)\varphi^\delta. \ee Then our quasi-2D equation II
(\ref{gpe2d2}) can be written as \be\label{gpe2d2r} i \p_t
\phi=\H_{\bx}^V\phi+g_1(\phi)+\tilde{g}_2(\phi), \ee where
\be\tilde{g}_2(\phi)=-\frac{3\lambda}{2}\phi(
\p_{\bn_\perp\bn_\perp} -n_3^2\Delta)(-\Delta)^{-1/2}(|\phi|^2).\ee
We denote the pairing of $X_2$ and its dual $X_2^\ast$ by
$\langle,\rangle_{X_2,X_2^\ast}$ as \be \langle f_1,f_2
\rangle_{X_2,X_2^\ast}=\mbox{Re}\int_{\Bbb
R^2}f_1(\bx)\bar{f}_2(\bx)\,d\bx, \ee where $\mbox{Re}(f)$ denotes
the real part of $f$. Using the results in Theorem \ref{thm1dy} and
\cite{Cazen}, we see that there exists a unique maximal solution
$\varphi^\delta\in C([-T_{\rm min}^\delta,T_{\rm max}^\delta],
X_2)\cap C^1([-T_{\rm min}^\delta,T_{\rm max}^\delta],X_2^\ast)$.
Here maximal means that if either $t\uparrow T_{\rm max}^\delta$  or
$t\downarrow -T_{\rm min}^\delta$,
$\|\phi^\delta(t)\|_{X_2}\to\infty$. We want to show that as
$\delta\to 0^+$, $\phi^\delta$ will converge to a solution of
equation (\ref{gpe2d2}).

 First, we show that $T_{\rm min}^\delta=+\infty$ and
 $T_{\rm max}^\delta=+\infty$. The energy conservation for (\ref{eq:appro}) is
\be E_\delta(t):=\frac12\|\nabla
\phi^\delta\|_2^2+\frac12\beta_0\|\phi^\delta\|_4^4+ \int_{\Bbb
R^2}V_2(\bx)|\phi^\delta|^2d\bx+E_{\rm dip}^\delta(t)\equiv
E_\delta(0), \qquad t\ge0, \ee where \be E_{\rm
dip}^\delta(t)=-\frac{3\lambda}{4}\int_{\Bbb R^2}|\phi^\delta|^2(
\p_{\bn_\perp\bn_\perp} -n_3^2\Delta)\varphi^\delta d\bx, \qquad
t\ge0. \ee Similar computation as in Lemma \ref{lem4} confirms that
$E_{\rm dip}^\delta\ge0$.
  Hence energy conservation will imply that $\|\phi^\delta(t)\|_{X_2}<\infty$
  for all $t$, i.e. $T^\delta_{\rm max}=T^\delta_{\rm min}=+\infty$.

We notice that \be X_2\hookrightarrow H^1\hookrightarrow
L^2\hookrightarrow H^{-1}\hookrightarrow X_2^\ast,\ee where $H^{-1}$
is viewed as the dual of $H^1$. Consider a bounded time interval
$I=[-T,T]$, it follows from energy conservation that there exists a
constant $C_1(\phi_0)>0$ such that \be\label{bound1}
\|\phi^\delta\|_{C([-T,T];X_2)}\leq C_1(\phi_0). \ee Moreover, Lemma
\ref{lem1} and Remark \ref{rmk} would imply \be \|\phi^\delta(
\p_{\bn_\perp\bn_\perp} -n_3^2\Delta)\varphi^\delta\|_q\leq
C\|\phi^\delta\|_{q^\ast}\|\nabla|\phi^\delta|^2\|_p \leq
C\|\phi^\delta\|_{q^\ast}\|\phi^\delta\|_{2p/(2-p)}\|\nabla\phi^\delta\|_2,
\ee
  for $q,p\in (1,2)$, $\frac{1}{q^\ast}+\frac{1}{p}=\frac1q$. Then we have
  \be\label{bound2}
 \| \phi^\delta\|_{C^1([-T,T];X_2^\ast)}\leq C_2(\phi_0).
  \ee
Thus, from (\ref{bound1}) and (\ref{bound2}), there exists a
sequence $\delta_n\to0^+$ ($n=1,2,\ldots,$) and a function $\phi\in
L^\infty([-T,T];X_2) \cap W^{1,\infty}([-T,T];X_2^\ast)$
\cite{Cazen},  such that \be\label{weakH1}
\phi^{\delta_n}(t)\rightharpoonup\phi(t)\quad \mbox{in }
X_2,\,\mbox{for all }t\in [-T,T]. \ee For each $t\in[-T,T]$, due to
the mass conservation of the equation (\ref{eq:appro}), we know
$\|\phi^{\delta_n}(t)\|_2=\|\phi_0\|_2$, by a similar proof in
Theorem \ref{thm1}, the weak convergence of $\phi^{\delta_n}(t)$ in
$X_2$ would imply that $\phi^{\delta_n}(t)$ converges strongly in
$L^2$, which is a consequence of the fact that
$V_2(\bx)$ is a confining potential. So,
$\lim\limits_{n\to\infty}\|\phi^{\delta_n}(t)\|_2=\|\phi(t)\|_2$,
and it turns out that \cite{Cazen} \be\label{strongL2}
\phi^{\delta_n}\to\phi,\quad \mbox{ in } C([-T,T];L^2(\Bbb R^2)).
\ee In view of (\ref{weakH1}), (\ref{strongL2}) and the
Gagliardo-Nirenberg's inequality, we obtain \be\label{strongLp}
\phi^{\delta_n} \to\phi,\quad \mbox{ in } C([-T,T];L^p(\Bbb R^2)),
\quad \mbox{ for all } p\in[2,\infty). \ee We now try to say that
$\phi$ actually solves equation (\ref{gpe2d2}). For any function
$\psi(\bx)\in X_2$ and $f(t)\in C_c^\infty([-T,T])$, from equation
(\ref{eq:appro}), we have \be\label{weakf} \int_{-T}^T\left[\langle
i\phi^{\delta_n},\psi\rangle_{X_2,X_2^\ast}f^\prime(t) +\langle
\H_{\bx}^V\phi^{\delta_n}+g_1(\phi^{\delta_n})+g_2(\phi^{\delta_n}),
\psi\rangle_{X_2,X_2^\ast}f(t)\right]dt=0. \ee Recalling
$|g_1(u)-g_1(v)|\leq C(|u|^2+|v|^2)|u-v|$, (\ref{strongLp}) implies
that \cite{Cazen} for all $t\in[-T,T]$ \be\label{sl:nonlinear}
g_1(\phi^{\delta_n}(t))\to g_1(\phi(t)),\quad\hbox{in } L^\rho(\Bbb
R^2) \text{ for some } \rho\in[1,\infty), \ee
 \be\label{g1} \langle g_1(\phi^{\delta_n}(t)),\psi(t) \rangle_{X_2,X_2^\ast}
 \to \langle g_1(\phi(t)),\psi(t)\rangle_{X_2,X_2^\ast}.\ee
 For $g_2(\phi^{\delta_n})$, consider  $\varphi^{\delta_n}(\bx,t)$,
 noticing
  $\p_{\alpha}\varphi^{\delta_n}=T_\alpha^{\delta_n}(|\phi^{\delta_n}|^2)$ ($\alpha=x,y$)
  (defined in Lemma \ref{converge:2D}),  we have proven in  Lemma
  \ref{converge:2D}  that
  $T_\alpha^{\delta_n}$ is uniformly bounded from $L^p$ to $L^p$  and
  \be
  T_\alpha^{\delta_n}(|\phi(t)|^2)\to R_\alpha(|\phi(t)|^2)=\p_{\alpha}(-\Delta)^{-1/2}(|\phi(t)|^2)
   \hbox{ in } L^p(\Bbb R^2), \qquad p\in(1,\infty),\quad
   \delta_n\to0^+\,.
  \ee
 Rewriting
  \be
  T_\alpha^{\delta_n}(|\phi^{\delta_n}(t)|^2)=T_\alpha^{\delta_n}(|\phi^{\delta_n}(t)|^2-
  |\phi(t)|^2)+T_\alpha^{\delta_n}(|\phi(t)|^2),
  \ee
  recalling the fact (\ref{strongLp}), we immediately have
  \be
  T_\alpha^{\delta_n}(|\phi^{\delta_n}(t)|^2)\to R_\alpha(|\phi(t)|^2) \mbox{ in } L^p(\Bbb
  R^2),\quad
  \mbox{ for some  } p\in(1,\infty),
  \ee
  which is actually
  \be\label{con:deri}
  \p_{\alpha}\varphi^{\delta_n}(t)\to
  \p_{\alpha}\left((-\Delta)^{-1/2}|\phi(t)|^2\right),\quad
  \hbox{ in } L^p(\Bbb R^2),\quad  \hbox{ for some  } p\in(1,\infty).
  \ee
 Hence, integration by parts, for $\alpha^\prime=x,y$,
 \beas
 \langle \phi^{\delta_n}(t)\p_{\alpha\alpha^\prime}\varphi^{\delta_n}(t),\psi(t) \rangle_{X_2,X_2^\ast}&=&
 \mbox{Re}\int_{\Bbb R^2}\phi^{\delta_n}(t)\p_{\alpha\alpha^\prime}\varphi^{\delta_n}(t) \bar{\psi}(t)d\bx\\
 &=&-\hbox{Re}\int_{\Bbb R^2}\p_{\alpha}\varphi^{\delta_n}(t) (\p_{\alpha^\prime}\phi^{\delta_n}(t)
 \bar{\psi}(t)+\phi^{\delta_n}(t)\overline{\p_{\alpha^\prime}\psi}(t))d\bx,
 \eeas
passing to the limit as $n\to\infty$,
\beas
 \lim\limits_{n\to\infty}\langle \phi^{\delta_n}(t)\p_{\alpha\alpha^\prime}\varphi^{\delta_n}(t),\psi(t) \rangle_{X_2,X_2^\ast}
 &=&-\mbox{Re}\int_{\Bbb R^2}R_\alpha(|\phi(t)|^2) (\p_{\alpha^\prime}\phi(t)\bar{\psi}(t)+\phi(t)\overline{\p_{\alpha^\prime}\psi}(t))d\bx\\
 &=&\langle \phi(t)\p_{\alpha\alpha^\prime}(-\Delta)^{-1/2}(|\phi(t)|^2),\psi(t) \rangle_{X_2,X_2^\ast},
 \eeas in view of (\ref{con:deri}) and (\ref{weakH1}), we obtain
\be\label{sl:nonlocal}
\lim\limits_{n\to\infty}\langle g_2(\varphi^{\delta_n}(t)),\psi(t) \rangle_{X_2,X_2^\ast}
=\langle \tilde{g}_2(\phi(t)),\psi(t) \rangle_{X_2,X_2^\ast}.
\ee
Combining the above results and (\ref{g1}) together, sending $n\to\infty$,
dominated convergence theorem will yield
\beas
\int_{-T}^T\left[\langle i\phi,\psi\rangle_{X_2,X_2^\ast}f^\prime(t)+\langle
\H_{\bx}^V\phi+g_1(\phi)+\tilde{g}_2(\phi),\psi\rangle_{X_2,X_2^\ast}f(t)\right]dt=0,
\eeas
which proves that
\be\label{weaksl}
i \p_t
\phi=\H_{\bx}^V\phi+g_1(\phi)+\tilde{g}_2(\phi), \mbox{ in } X_2^\ast, \quad \mbox{ a.a. } t\in[-T,T],
\ee
with $\phi(t=0)=\phi_0$, and $\phi\in L^\infty([-T,T];X_2)\cap W^{1,\infty}([-T,T];X_2^\ast)$.
Moreover, by lower semi-continuity of $X_2$ norm, (\ref{sl:nonlinear}) and (\ref{sl:nonlocal}),
the energy $\tilde{E}_{2D}$ (\ref{ener2d2}) satisfies
\be
\tilde{E}_{2D}(\phi(t))\leq \tilde{E}_{2D}(\phi_0).
\ee
It is easy to see that we can choose $T=\infty$.\end{proof}

\medskip

If  the uniqueness of the $L^\infty([-T,T];X_2)\cap W^{1,\infty}([-T,T];X_2^\ast)$
solution to the quasi-2D equation II (\ref{gpe2d2}) is known, we can prove that the
solution constructed above in the Theorem \ref{dyn2d2} is actually
$C([-T,T];X_2)\cap C^1([-T,T];X_2^\ast)$  and conserves the energy.

\medskip

Next, we discuss possible finite time blow-up for the continuous
solutions of the quasi-2D equation II (\ref{gpe2d2}). To this
purpose, the following assumptions are introduced:

(A) Assumption on the trap and coefficient of the cubic term, i.e.
$V_2(\bx)$ satisfies $3V_2(\bx)+ \bx\cdot \nabla V_2(\bx)\ge0$,
$\frac{\beta-\lambda+3\lambda
n_3^2}{\sqrt{2\pi}\,\vep}\ge-\frac{C_b}{\|\phi_0\|_2^2}$, with
$\phi_0$ being the initial data of equation (\ref{gpe2d2});

(B) Assumption on the trap and coefficient of the nonlocal term, i.e. $V_2(\bx)$
satisfies $2V_2(\bx)+ \bx\cdot \nabla V_2(\bx)\ge0$, $\lambda=0$ or $\lambda>0$ and $n_3^2\ge\frac12$.

\begin{theorem}(Finite time blow-up)
For any initial data $\phi(\bx,t=0)=\phi_0(\bx)\in X_2$ with $\int_{\Bbb
R^2}|\bx|^2|\phi_0(\bx)|^2\,d\bx<\infty$,
if conditions (B1), (B2) and (B3) with constant $C_b$ being replaced by
$C_b/\|\phi_0\|_2^2$ are not satisfied,
let $\phi:=\phi(\bx,t)$ be a
$C([0,T_{max}),X_2)$ solution of the problem (\ref{gpe2d2}) with
$L^2$ norm and energy conservation, then there exists finite time
blow-up, i.e., $T_{\mbox{max}}<\infty$, if one of the following
condition holds:

(i) $\tilde{E}_{2D}(\phi_0)<0$, and either Assumption (A) or (B) holds;

(ii) $\tilde{E}_{2D}(\phi_0)=0$ and ${\rm Im}\left(\int_{\Bbb
R^2}\bar{\phi}_0(\bx)\ (\bx\cdot\nabla\phi_0(\bx))\,d\bx\right)<0$,
and either Assumption (A) or (B) holds;

(iii) $\tilde{E}_{2D}(\phi_0)>0$, and ${\rm Im}\left(\int_{\Bbb R^2}
\bar{\phi}_0(\bx)\ (\bx\cdot\nabla\phi_0(\bx))\,d\bx\right)
<-\sqrt{3\tilde{E}_{2D}(\phi_0)}\|\bx\phi_0\|_{2}$ if Assumption (A)
holds, or ${\rm Im}\left(\int_{\Bbb R^2} \bar{\phi}_0(\bx)\
(\bx\cdot\nabla\phi_0(\bx))\,d\bx\right)
<-\sqrt{2\tilde{E}_{2D}(\phi_0)}\|\bx\phi_0\|_{2}$ if Assumption (B)
holds.
\end{theorem}

\begin{proof} Calculating  derivatives of  the variance defined in
(\ref{dtv001}),  for $\alpha=x,y$, we have \be
\frac{d}{dt}\sg_\alpha(t)=2\,{\rm Im}\left(\int_{\Bbb R^2}
\bar{\phi}\ \alpha\p_{\alpha}\phi\,d\bx\right), \qquad
t\ge0,\ee and \be\label{d2ap23}
\frac{d^2}{dt^2}\sg_\alpha(t)=\int_{\Bbb
R^2}\left[2|\p_{\alpha}\phi|^2+\beta_0
|\phi|^4+3\lambda|\phi|^2\alpha\p_{\alpha}(\p_{\bn_\perp\bn_\perp}-n_3^2\Delta)
\varphi-2\alpha|\phi|^2\p_{\alpha}V_2(\bx)\right]\,d\bx,\ee where
$\beta_0=\frac{\beta-\lambda+3\lambda n_3^2}{\sqrt{2\pi}\vep}$,
$(-\Delta)^{1/2}\varphi=|\phi|^2$. Writing $\rho=|\phi|^2$,
$\tilde{\varphi}=(\p_{\bn_\perp\bn_\perp}-n_3^2\Delta) \varphi$ and
noticing that $\rho$ is a real function,  by the Plancherel formula,
similar as Theorem \ref{blowup-1}, we get
\be \int_{\Bbb R^2}|\phi|^2\left(\bx\cdot\nabla
\tilde{\varphi}\right)\,d\bx =-\frac32\int_{\Bbb R^2}|\phi|^2
\tilde{\varphi}\,d\bx. \ee Hence, summing (\ref{d2ap23}) for
$\alpha=x$, $y$, and using the energy conservation,  if Assumption
(A) holds, we have
\begin{eqnarray}
\frac{d^2}{dt^2}\sg_V(t)&=&2\int_{\Bbb
R^2}\left(|\nabla\phi|^2+\beta_0|\phi|^4-\frac
94\lambda |\phi|^2\left(\p_{\bn_\perp\bn_\perp}-n_3^2\Delta\right)\varphi -|\phi|^2(\bx\cdot\nabla
V_2(\bx))\right)\,d\bx\nn\\
&=&6\tilde{E}_{2D}(\phi)-\int_{\Bbb R^2}(|\nabla\phi|^2+\beta_0|\phi|^4)\,d\bx-2\int_{\Bbb
R^2}|\phi|^2\left(3V_2(\bx)+\bx\cdot\nabla V_2(\bx)\right)\,d\bx\nn\\
&\leq&6\tilde{E}_{2D}(\phi(\cdot,t))\leq 6\tilde{E}_{2D}(\phi_0), \qquad t\ge0.
\end{eqnarray}
Thus,
$$
\sg_V(t)\leq 3\tilde{E}_{2D}(\phi_0)t^2+\sg_V^\prime(0)t+\sg_V(0), \qquad t\ge0,
$$
and the conclusion follows as in Theorem \ref{blowup-1}.
 If Assumption (B) holds, the energy contribution of the nonlocal part is non-positive and we have
\begin{eqnarray}
\frac{d^2}{dt^2}\sg_V(t)&=&2\int_{\Bbb
R^2}\left(|\nabla\phi|^2+\beta_0|\phi|^4-\frac
94\lambda |\phi|^2\left(\p_{\bn_\perp\bn_\perp}-n_3^2\Delta\right)\varphi -|\phi|^2(\bx\cdot\nabla
V_2(\bx))\right)\,d\bx\nn\\
&=&4\tilde{E}_{2D}(\phi)-\frac{3\lambda}{2}\int_{\Bbb R^2}|\phi|^2\tilde{\varphi}\,d\bx-2\int_{\Bbb
R^2}|\phi|^2\left(2V_2(\bx)+\bx\cdot\nabla V_2(\bx)\right)\,d\bx\nn\\
&\leq&4\tilde{E}_{2D}(\phi(\cdot,t))\leq 4\tilde{E}_{2D}(\phi_0), \qquad t\ge0,
\end{eqnarray}
and the conclusion follows in a similar way as  the Assumption (A) case.
\end{proof}

\section{Results for quasi-1D equation}
In this section, we prove the existence and uniqueness of
the ground state for quasi-1D equation (\ref{gpe1d}) and
establish the well-posedness for dynamics.

\subsection{Existence and uniqueness of ground state}
Associated to the quasi-1D equation (\ref{gpe1d}), the energy is
\be\label{ener1d} E_{1D}(\Phi)=\int_{\Bbb
R}\left[\frac12|\p_{z}\Phi|^2+V_1(z)|\Phi|^2+\frac12\beta_{1D}|\Phi|^4
+\frac{3\lambda(1-3n_3^2)}{16\sqrt{2\pi}\,\vep}|\Phi|^2
\varphi\right]\,dz, \ee
 where $\beta_{1D}=\frac{\beta+\frac12\lambda(1-3n_3^2)}{2\pi\vep^2}$ and
\be\label{kernel1d} \varphi(z)=\p_{zz}(U_\vep^{1D}*|\Phi|^2),\quad
U_\vep^{1D}(z)=
\frac{2e^{\frac{z^2}{2\vep^2}}}{\sqrt{\pi}}\int_{|z|}^\infty
e^{-\frac{s^2}{2\vep^2}}\,ds. \ee Again, the ground state $\Phi_g\in
S_1$ of the equation (\ref{gpe1d}) is defined as the minimizer of
the nonconvex minimization problem: \be \mbox{Find } \Phi_g\in
S_1,\quad\mbox{such that }E_{1D}(\Phi_g)=\min\limits_{\Phi\in S_1}\;
E_{1D}(\Phi). \ee

For the above ground state, we have  the following results.

\begin{theorem}\label{thm1''}(Existence and uniqueness of ground state) Assume
$0\leq V_1(z)\in L^\infty_{\rm loc}(\Bbb R)$ and
$\lim\limits_{|z|\to\infty}V_1(z)=\infty$, for any parameter
$\beta$, $\lambda$ and $\vep$, there exists a ground state
$\Phi_g\in S_1$ of the  quasi-1D equation
(\ref{gpe1d})-(\ref{poisson1d}), and the positive ground state
$|\Phi_g|$ is unique under one of the following conditions:

(C1)  $\lambda(1-3n_3^2)\ge0$ and $\beta-(1-3n_3^2)\lambda\ge0$;

 (C2)  $\lambda(1-3n_3^2)<0$ and  $\beta+\frac{\lambda}{2}(1-3n_3^2)\ge0$.

\noindent Moreover, $\Phi_g=e^{i\theta_0}|\Phi_g|$ for some constant
$\theta_0\in\Bbb R$.
\end{theorem}

To complete the proof, we first study the property of the convolution
kernel  $U_\vep^{1D}$ (\ref{poisson1d}).

\begin{lemma}\label{lem5}(Kernel $U_\vep^{1D}$(\ref{poisson1d}))
For any $f(z)$ in the Schwartz space ${\mathcal{S}}(\Bbb R)$, we
have \be\label{u1dcov}
\widehat{U_\vep^{1D}*f}(\xi)=\hat{f}(\xi)\widehat{U_\vep^{1D}}(\xi)=
\frac{\sqrt{2}\,\vep\hat{f}(\xi)}{\sqrt{\pi}}\int_0^\infty\frac{e^{-\vep^2
s/2}}{|\xi|^2+s}ds, \quad \xi\in{\Bbb R}. \ee Hence \be
\|\p_{zz}(U_\vep^{1D}*f)\|_2\leq\frac{2\sqrt{2}}{\sqrt{\pi}\,\vep}\|f\|_2.
\ee
\end{lemma}

\begin{proof} For any $f(z)\in\mathcal{S}({\Bbb R})$,  rewrite the
kernel as \cite{Bao2}\be\label{eq:u1d}
U_\vep^{1D}(z)=\frac{\sqrt{2}\,\vep}{\sqrt{\pi}}\int_{\Bbb R^4}
\frac{w_\vep^2(x,y)w_\vep^2(x^\prime,y^\prime)}
{\sqrt{z^2+(x-x^\prime)^2+(y-y^\prime)^2}}\,dxdydx^\prime dy^\prime,
\nonumber \ee
applying Fourier transform to both sides and using the Plancherel
formula as in Lemma \ref{lem1}, we obtain \be
\widehat{U_\vep^{1D}}(\xi)=\frac{\sqrt{2}\,\vep}{\pi^{3/2}}\int_{\Bbb
R^2}\frac{|\widehat{w_\vep^2}(\xi_1,\xi_2)|^2}{\xi^2+\xi_1^2+\xi_2^2}\,d\xi_1\,d\xi_2,
\qquad \xi\in{\Bbb R}. \ee Then a direct computation  yields the
conclusion. \end{proof}

\begin{lemma}\label{lem6} For the energy  $E_{1D}(\cdot)$ in (\ref{ener1d}), we have

(i) For any $\Phi\in S_1$, denote $\rho(z)=|\Phi(z)|^2$, then  \be
E_{1D}(\Phi)\geq E_{1D}(|\Phi|)=E\left(\sqrt{\rho}\right), \qquad
\forall \Phi\in S_1,\ee so the ground state $\Phi_g$ of
(\ref{ener1d}) is of the form $e^{i\theta_0}|\Phi_g|$ for some
constant $\theta_0\in \Bbb R$.

(ii) $E_{1D}$ is bounded below.

 (iii) If   condition (C1) or (C2) in  Theorem \ref{thm1''} holds,  $E_{1D}(\sqrt{\rho})$
 is strictly convex in $\rho$.
\end{lemma}
\begin{proof} Part (i) is similar to that in Lemma \ref{lem1}. Part (ii) is
 well-known, once we notice the property of kernel $U_\vep^{1D}$ (cf. Lemma
 \ref{lem5}) and the Sobolev inequality in one dimension,
\be
\|f\|^2_\infty\leq \|f\|_2\|f^\prime\|_2.
\ee

(iii) We come to the convexity of $E_{1D}(\sqrt{\rho})$. Following
Lemma \ref{lem2}, we only need  to consider the functional \be
H_{1D}(\rho)=\int_{\Bbb
R}\left[\frac{\beta+\lambda(1-3n_3^2)/2}{4\pi\vep^2}\rho^2
+\frac{3\lambda(1-3n_3^2)}{16\sqrt{2\pi\vep^2}}\rho\p_{zz}(U_\vep^{1D}*\rho)\right]\,dz.
\ee Then under condition (C1) or (C2), using the Plancherel formula
and Lemma \ref{lem5}, after similar computation as in Lemma
\ref{lem1}, we would have $H_{1D}(\rho)\ge0$. For arbitrary
$\sqrt{\rho_1},\sqrt{\rho_2}\in S_1$ and $\theta\in[0,1]$, denote
$\rho_{_\theta}=\theta\rho_1+(1-\theta)\rho_2$, then
$\sqrt{\rho_{_\theta}}\in S_1$ and \be \theta
H_{1D}(\rho_1)+(1-\theta)
H_{1D}(\rho_2)-H_{1D}(\rho_{_\theta})=\theta(1-\theta)H_{1D}(\rho_1-\rho_2)\ge0,
\ee which proves the convexity. \end{proof}

\bigskip

{\bf{Proof of Theorem \ref{thm1''}}:} The uniqueness  follows from the strict
 convexity in Lemma \ref{lem6}. The
existence part is similar  as  Theorem  \ref{thm1} and we omit it here for brevity. \hfill$\Box$

\subsection{Well-posedness for the Cauchy problem}
Concerning the Cauchy problem, Lemma \ref{lem5} shows that the
nonlinearity in quasi-1D equation (\ref{gpe1d}) is almost  a cubic
nonlinearity, while the same property has been observed for
quasi-2D equation I (\ref{gpe2d})-(\ref{u2d1}). Hence similar
results as Theorem \ref{thm1dy} can be obtained for equation
(\ref{gpe1d})  and we omit the proof here.

\begin{theorem}
(Well-posedness for Cauchy problem) Suppose the real-valued trap
potential satisfies $V_1(z)\ge0$ for $z\in{\Bbb R}$ and $V_1(z)\in
C^\infty(\Bbb R)$, $D^k V_1(z)\in L^\infty(\Bbb R)$ for all
integers $k\ge 2$, for any initial data
$\phi(z,t=0)=\phi_0(z)\in X_1$,
 there exists
 a unique solution
$\phi\in C\left([0,\infty),X_1\right)\cap C^1\left([0,\infty),X_1^\ast\right)$
to the Cauchy problem of equation (\ref{gpe1d}).
\end{theorem}

\section{Convergence rate of dimension reduction}
\setcounter{equation}{0}

In this section, we  discuss the dimension reduction of 3D GPPS to
lower dimensions. Inspired by the previous work of Ben Abdallah et
al. \cite{bamsw,bacm}  for GPE without the dipolar term (i.e.
$\lambda=0$) and \cite{baca,bafm} for Schr\"odinger-Poisson systems,
we are going to find a limiting $\vep$-independent
equation as $\vep\to0^+$. Thus in quasi-2D equation I (\ref{gpe2d}),
II (\ref{gpe2d2}) and quasi-1D equation (\ref{gpe1d}), we have to
consider the coefficients to be $O(1)$.
 The existence of the global solution for  the full 3D system
(\ref{gpe})-(\ref{poisson})  has been proven in \cite{Carles,Bao1}
when $\beta\ge0$ and $\lambda\in[-\frac12\beta,\beta]$, hence we
would expect the limiting equation in lower dimensions valid in a
similar regime. Thus in lower dimensions,
 we require that in the quasi-2D case, $\beta=O(\vep)$, $\lambda=O(\vep)$, and in
 the quasi-1D case, $\beta=O(\vep^2)$, $\lambda=O(\vep^2)$, i.e. we are considering
 the weak interaction regime, then we would get an $\vep$-independent limiting equation.
 In this regime, we will see that  GPPS will reduce to a regular GPE in lower dimensions.

 \subsection{Reduction to 2D}
 We consider the weak interaction regime, i.e.,   $\beta\to \vep\beta$, $\lambda\to\vep\lambda$.
In {\sl{Case I}} (\ref{case1}),  for full 3D GPPS
(\ref{gpe})-(\ref{poisson}), introduce the re-scaling $z\to \vep z$,
$\psi\to \vep^{-1/2}\psi^\vep$ which preserves the normalization,
then
 \bea\label{rescalgpe2}
 i\p_t\psi^\vep(\bx,z,t)=\left[\H_{\bx}^V+\frac{1}{\vep^2}\H_z+(\beta-\lambda)|\psi^\vep|^2-3
 \vep\lambda\p_{\bn_\vep\bn_\vep}\varphi^\vep\right] \psi^\vep,\quad (\bx,z)\in \Bbb
 R^3, \ t>0,\quad
 \eea
 where $\bx=(x,y) \in \Bbb R^2$ and
 \bea
 &&\H_{\bx}^V=-\frac{1}{2}(\p_{xx}+\p_{yy})+V_2(x,y),\qquad \H_z=-\frac12\p_{zz}+\frac{z^2}{2},\\
 &&\bn_\vep=(n_1,n_2,n_3/\vep),\qquad \p_{\bn_\vep}=\bn_\vep\cdot\nabla,\qquad
 \p_{\bn_\vep\bn_\vep}=\p_{\bn_\vep}(\p_{\bn_\vep}),\\
 &&(-\p_{xx}-\p_{yy}-\frac{1}{\vep^2}\p_{zz})\varphi^\vep=\frac{1}{\vep}|\psi^\vep|^2,\qquad
 \lim\limits_{|\bx|\to \infty}\varphi^\vep(\bx)=0.\label{rescalps}
 \eea
 It is well-known that $\H_z$ has eigenvalues $\mu_k=k+1/2$ with corresponding
 eigenfunction $w_k(z)$ ($k=0,1,\ldots$), where $\{w_k\}_{k=0}^\infty$ forms an
 orthornormal   basis  of $L^2(\Bbb R)$ \cite{DF,GS}, specifically, $w_0(z)=
 \frac{1}{\pi^{1/4}}e^{-z^2/2}$.  Following \cite{bamsw}, it is convenient to
 consider the initial data concentrated on the ground mode of $\H_z$, i.e.,
 \be\label{init}
 \psi^\vep(\bx,z,0)=\phi_0(\bx)w_0(z),\quad \phi_0\in X_2 \mbox{ and } \|\phi_0\|_{L^2(\Bbb R^2)}=1.
 \ee

In {\sl{Case I}} (\ref{case1}),
 when $\vep\to0^+$, quasi-2D equation I (\ref{gpe2d}), II (\ref{gpe2d2})
 will yield  an $\vep$-independent equation in the weak interaction regime,
 \be\label{eq:2d:weak}
 i\p_t \phi(\bx,,t)=\H_{\bx}^V\phi+\frac{\beta-(1-3n_3^2)\lambda}{\sqrt{2\pi}}
 |\phi|^2\phi,\quad \bx=(x,y)\in\Bbb R^2,
 \ee
 with initial condition $\phi(\bx,0)=\phi_0(\bx)$.
 We follow the ideas in \cite{baca,bamsw,bacm} to show the convergence from the
 3D GPPS to the 2D approximation.  First, let us state the main result.

\begin{theorem}\label{redthm1}(Dimension reduction to 2D) Suppose $V_2$ satisfies
condition (\ref{cond:v2}), $-\frac{\beta}{2}\leq \lambda \leq\beta$
and $\beta\ge0$, let $\psi^\vep\in C([0,\infty);X_3)$ and
 $\phi\in C([0,\infty);X_2)$ be the unique solutions of equations (\ref{rescalgpe2})-(\ref{init})
 and (\ref{eq:2d:weak}), respectively, then for any $T>0$, there exists $C_{T}>0$ such that
\be \left\|\psi^\vep(\bx,z,t)-e^{-i\frac{\mu_0
t}{\vep^2}}\phi(\bx,t)w_0(z)\right\|_{L^2(\Bbb R^3)}\leq
C_{T}\,\vep, \qquad \forall t\in [0,T]. \ee
\end{theorem}

 Under the assumption, we have the global existence of $\psi^\vep$ \cite{Bao1,Carles} as well as
 $\phi$ \cite{bamsw,Cazen}.
 Define the projection operator onto the ground mode $\H_z$ by
 \be\label{psiproj}
\Pi \psi^\vep(\bx,z,t)=e^{-i\mu_0t/\vep^2}\phi^\vep(\bx,t)w_0(z),
\qquad (\bx,z)\in{\Bbb R}^3, \quad t\ge0,\ee where
\be\label{phieps}\phi^\vep(\bx,t)=e^{i\mu_0t/\vep^2}\int_{\Bbb
R}\psi^\vep(\bx,z,t)w_0(z)dz, \qquad \bx=(x,y)\in{\Bbb R}^2, \quad
t\ge0.\ee Since the space $(\bx,z)\in{\Bbb R}^3$ is anisotropic, we introduce the
$L_{z}^pL_{\bx}^q$ space by the norm \be
\|f\|_{(p,q)}:=\|f\|_{L^p_zL_{\bx}^q}=\|\,\|f(\cdot,z)\|_{L_{\bx}^q}\|_{L_z^p},\quad
p,q\in [1,\infty]. \ee The corresponding anisotropic Sobolev
inequalities are available \cite{bamsw}.

\begin{lemma}\label{redlem1}(Uniform bound) Let $\psi^\vep$ and $\phi$ be the solutions of
(\ref{rescalgpe2}) and  (\ref{eq:2d:weak}), respectively,
$\phi^\vep$ be defined in (\ref{phieps}),
$\lambda\in[-\frac{\beta}{2},\beta]$ and
 $\beta\ge0$, we have
\be
\psi^\vep\in L^\infty((0,\infty),H^1(\Bbb R^3)),\quad \phi,\phi^\vep\in L^\infty((0,T), H^1(\Bbb R^2)),
\ee
with  uniform bound in $\vep$. Moreover, for $p\in[2,\infty]$,
\be
 \|\psi^\vep-\Pi\psi^\vep\|^2_{L^2(\Bbb R^3)}+\|\p_z(\psi^\vep-\Pi\psi^\vep)
 \|_{L^2(\Bbb R^3)}^2\leq C\vep^2,\qquad
 \|\psi^\vep-\Pi\psi^\vep\|_{(p,2)}\leq C\vep,
 \ee
 with $C$ depending on $\|\phi_0\|_{X_2}$, uniform in time $t$.
\end{lemma}

 \begin{proof}
From energy conservation of equation (\ref{rescalgpe2}), we know
that  the following energy is constant, \be
E(t):=(\H_{\bx}^V\psi^\vep(t),\psi^\vep(t))+\frac{1}{\vep^2}(\H_z\psi^\vep(t),\psi^\vep(t))
+\frac{\beta-\lambda}{2}\|\psi^\vep\|_{L^4(\Bbb
R^3)}^4+\frac{3\lambda\vep^2}{2}\|\p_{\bn_\vep}
\nabla\varphi^\vep(t)\|_{L^2(\Bbb R^3)}^2,\nonumber \ee where
$(\cdot,\cdot)$ denotes the standard $L^2$ inner product in $\Bbb
R^3$. Using standard $L^p$ estimates for Poisson equation
(\ref{rescalps}), we have
$\|\p_{\bn_\vep}\nabla\varphi^\vep(t)\|_{L^2(\Bbb R^3)}\leq
\frac{1}{\vep}\|\psi^\vep(t)\|_{L^4(\Bbb R^3)}^2$, which implies \be
\frac{\beta-\lambda}{2}\|\psi^\vep\|_{L^4(\Bbb
R^3)}^4+\frac{3\lambda\vep^2}{2}\|\p_{\bn_\vep}
\nabla\varphi^\vep(t)\|_{L^2(\Bbb R^3)}^2\ge0,\quad \mbox{ and }\
E(0)=\frac{\mu_0}{\vep^2}+C_0, \ee where $C_0$ depends on
$\|\phi_0\|_{X_2}$. Writing
$\psi^\vep(t)=\sum\limits_{k=0}^\infty\phi_k(\bx,t)w_k(z)$, and
using the $L^2$ conservation
$\sum\limits_{k=0}^\infty\|\phi_k(t)\|_{L^2(\Bbb R^2)}^2=1$, we can
deduce from the energy conservation that \bea
 \frac{\mu_0}{\vep^2}+C_0&\ge&(\H_{\bx}^V\psi^\vep(t),\psi^\vep(t))+
 \frac{1}{\vep^2}(\H_z\psi^\vep(t),\psi^\vep(t))
=(\H_{\bx}^V\psi^\vep(t),\psi^\vep(t))+\frac{1}{\vep^2}\sum\limits_{k=0}^\infty\mu_k
\|\phi_k(t)\|_{L^2(\Bbb R^2)}^2\nonumber\\
&=&(\H_{\bx}^V\psi^\vep(t),\psi^\vep(t))+\frac{1}{\vep^2}\sum\limits_{k=1}^\infty(\mu_k-\mu_0)
\|\phi_k(t)\|_{L^2(\Bbb R^2)}^2+\frac{\mu_0}{\vep^2}.\nonumber
\eea
Hence,
\bea
&&\|\p_{x}\psi^\vep\|^2_{L^2(\Bbb R^3)}+\|\p_{y}\psi^\vep\|^2_{L^2(\Bbb R^3)}\leq
(\H_{\bx}^V\psi^\vep(t),\psi^\vep(t))\leq C_0,\\
&&\|\p_z\psi^\vep\|^2_{L^2(\Bbb R^3)}\leq (\H_z\psi^\vep,\psi^\vep)\leq \mu_0+C_0\vep^2,\\
&&\|\psi^\vep-\Pi\psi^\vep\|^2_{L^2(\Bbb R^3)}\leq \frac{1}{\mu_1-\mu_0}
\sum\limits_{k=1}^\infty(\mu_k-\mu_0)\|\phi_k(t)\|_{L^2(\Bbb R^2)}^2\leq 2C_0\vep^2,\\
&&\|\p_z(\psi^\vep-\Pi\psi^\vep)\|^2_{L^2(\Bbb R^3)}\leq \sum\limits_{k=1}^\infty
\frac{\mu_k}{\mu_k-\mu_0}(\mu_k-\mu_0)\|\phi_k(t)\|_{L^2(\Bbb R^2)}^2\leq \frac32C_0\vep^2.
\eea
Estimate on $\|\psi^\vep-\Pi\psi^\vep\|_{(p,2)}$ follows from Sobolev embedding.
\end{proof}

 We also need the following Strichartz estimates for the unitary group $e^{-it\H_{\bx}^V}$,
 which is valid when $V_2$ satisfies condition (\ref{cond:v2}) \cite{Cazen}.

\begin{definition} In two dimensions, let $q^\prime$, $r^\prime$ be the conjugate
 index of $q$ and $r$ ($1\leq q,r\leq \infty$) respectively, i.e. $1=1/q^\prime+1/q
 =1/r^\prime+1/r$, we call the pair $(q,r)$  admissible and $(q^\prime,r^\prime)$ conjugate admissible if
\be \frac{2}{q}=2\left(\frac{1}{2}-\frac{1}{r}\right),\quad 2\leq
r<\infty. \ee
\end{definition}
The following
estimates are established in \cite{Cazen,Carles,Stri}.

\begin{lemma}\label{stri}(Strichartz's estimates) Let $(q,r)$ be an admissible
pair and $(\gamma,\rho)$ be a conjugate admissible pair, $I\subset
{\Bbb R}$ be a bounded interval satisfying $0\in I$, then we have

(i) There exists a constant  $C$ depending on $I$ and $q$ such that
\be \left\|e^{-it\H_{\bx}^V}\varphi\right\|_{L^q(I,L^r(\Bbb R^2))}\leq
C(I,q) \|\varphi\|_{L^2(\Bbb R^2)}. \ee

(ii)  If $f\in L^{\gamma}(I,L^{\rho}(\Bbb R^2))$, there exists a
constant  $C$ depending on $I$, $q$ and $\rho$, such that \be
\left\|\int_{I\bigcap s\leq
t}e^{-i(t-s)\H_{\bx}^V}f(s)\,ds\right\|_{L^q(I,L^r(\Bbb R^2))}\leq
C(I,q,\rho) \|f\|_{L^{\gamma}(I,L^{\rho}(\Bbb R^2))}. \ee
\end{lemma}

Now, we are able to prove the theorem.

{\bf{Proof of Theorem \ref{redthm1}}:} In view of Lemma \ref{redlem1}, we can derive
\bea
\|\psi^\vep-e^{-i\frac{t\mu_0}{\vep^2}}\phi w_0(z)\|_{L^2(\Bbb R^3)}&\leq&\|\psi^\vep
-\Pi\psi^\vep\|_{L^2(\Bbb R^3)}+\|\Pi\psi^\vep-e^{-i\frac{t\mu_0}{\vep^2}}
\phi w_0(z)\|_{L^2(\Bbb R^3)}\nonumber\\
&\leq&C\vep+\|\phi^\vep(t)-\phi(t)\|_{L^2(\Bbb R^2)}.\label{L2error}
\eea Hence, we need to estimate the difference between $\phi^\vep$
and $\phi$. By (\ref{rescalgpe2}) and (\ref{rescalps}), we know
$\phi^\vep(\bx,t)$ in (\ref{phieps}) solves the following equation
\beas
&&i\p_t\phi^\vep=\H_{\bx}^V\phi^\vep+(\beta-\lambda+3n_3^3\lambda)e^{i\mu_0t/\vep^2}
\int_{\Bbb R}|\psi^\vep|^2\psi^\vep w_0(z)dz+\vep g^\vep,\\
&& g^\vep=e^{i\mu_0t/\vep^2}\int_{\Bbb R}
P_\vep(\varphi^\vep) \psi^\vep w_0(z)dz,\eeas
where the differential operator $P_\vep$ is defined as \be P_\vep(\varphi^\vep)=
-3\lambda((n_1^2-n_3^2)\p_{xx}+(n_2^2-n_3^2)\p_{yy}+2n_1n_2\p_{xy}+\frac{2}{\vep}
(n_1n_3\p_{xz}+n_2n_3\p_{yz}))\varphi^\vep.
\ee
Denote $\chi^\vep(\bx,t)=\phi^\vep-\phi$, noticing that $\|w_0\|_4^4=1/\sqrt{2\pi}$,
$\chi^\vep$ satisfies the following equation
\beas
&&i\p_t\chi^\vep=\H_{\bx}^V\chi^\vep+f^\vep_1+f^\vep_2+\vep g^\vep,\quad \chi^\vep(t=0)=0,\\
&&f^\vep_1=\frac{\beta-\lambda+3n_3^3\lambda}{\sqrt{2\pi}}(|\phi^\vep|^2\phi^\vep-|\phi|^2\phi),\\
&&f^\vep_2=(\beta-\lambda+3n_3^3\lambda)e^{i\mu_0
t/\vep^2}\int_{\Bbb R}\left(|\psi^\vep|^2\psi^\vep- e^{-i\mu_0
t/\vep^2}|\phi^\vep w_0|^2\phi^\vep w_0\right)w_0(z)dz. \eeas
Applying Strichartz estimates on bounded interval $[0,T]$ and
recalling that $(\infty,2)$ is an admissible pair, we can obtain \be
\|\chi^\vep\|_{L^\infty([0,T];L^2(\Bbb R^2))}\leq C
\left[\|f^\vep_1\|_{L^{\rho^\ast} ([0,T];L^\rho(\Bbb
R^2))}+\|f^\vep_2\|_{L^{\gamma^\ast}([0,T];L^\gamma(\Bbb R^2))}+
\vep\|g^\vep\|_{L^{q^\ast}([0,T];L^q(\Bbb R^2))}\right],\nonumber \ee
where $(\rho^\ast,\rho)$, $(\gamma^\ast,\gamma)$ and $(q^\ast,q)$
are some conjugate admissible pairs. By a similar argument in
\cite{bamsw}, we have the estimates for $f^\vep_1$ and $f^\vep_2$
which comes from the cubic nonlinearity, i.e. for appropriate
$\rho\in(1,2)$ and $\gamma\in(1,2)$, \be
\|f^\vep_1\|_{L^{\rho^\ast}([0,T];L^\rho(\Bbb R^2))}\leq C
\|\chi^\vep\|_{L^{\rho^\ast}([0,T];L^2(\Bbb R^2))},\quad
\|f^\vep_2\|_{L^{\gamma^\ast}([0,T];L^\gamma(\Bbb R^2))}\leq C\vep.
\ee
 The basic tools involved are the H\"{o}lder's inequality, Sobolev inequalities and the
 estimates in Lemma \ref{redthm1}, and we omit the proof of this part here for brevity. So,
 \be\label{red:int}
\|\chi^\vep\|_{L^\infty([0,T];L^2(\Bbb R^2))}\leq C
(\|\chi^\vep\|_{L^{\rho^\ast} ([0,T];L^\rho(\Bbb
R^2))}+\vep\|g^\vep\|_{L^{q^\ast}([0,T];L^q(\Bbb R^2))}+\vep). \ee

 Next, we shall estimate $g^\vep$.  Let
 $\varphi^\vep_1$ and $\varphi^\vep_2$  be the solution of the  rescaled Poisson equation
 (\ref{rescalps}) with $|\psi^\vep|^2$ replaced by
 $|\Pi\psi^\vep|^2$ and $|\psi^\vep|^2-|\Pi\psi^\vep|^2$, respectively,
 then $\varphi^\vep=\varphi_1^\vep+\varphi_2^\vep$ and we can rewrite
 \be g^\vep=J_1^\vep+J_2^\vep+J_3^\vep,\ee
 where \be  J_1^\vep=\int_{\Bbb R}P_\vep(\varphi_1^\vep)\phi^\vep w_0^2dz,\, J_2^\vep
 =e^{\frac{i\mu_0t}{\vep^2}}\int_{\Bbb R}P_\vep(\varphi^\vep)(\psi^\vep-\Pi\psi^\vep)w_0dz,\, J_3^\vep=
 e^{\frac{i\mu_0t}{\vep^2}}\int_{\Bbb R}P_\vep(\varphi_2^\vep)\Pi\psi^\vep w_0dz.
 \nonumber\ee
 For $J^\vep_1$, this one reduces to the quasi-2D equation I (\ref{gpe2d}),
 where we have that
 \be
 J^\vep_1=-3\lambda (\p_{\bn_\perp\bn_\perp}-n_3^2\Delta)\varphi_{2D}^\vep
 \phi^\vep,\mbox{ and }\varphi^\vep_{2D}=U^{2D}_\vep*|\phi^\vep|^2,
 \ee
 with $U^{2D}_\vep$ given in (\ref{u2d1}). In view of the property of $U^{2D}_\vep$ in Lemma
 \ref{lem1} and Remark \ref{rmk}, recalling $\phi^\vep\in L^\infty([0,\infty);H^1(\Bbb R^2))$,
  using H\"{o}lder's inequality and Sobolev inequality, we obtain
 \be
 \|J^\vep_1\|_p\leq \|P_\vep(\varphi_{2D}^\vep)\|_{p_1}\|\phi^\vep\|_{p_2}\leq
 C\|\nabla|\phi^\vep|^2\|_{p_1} \|\phi^\vep\|_{p_2}\leq C,
 \ee
 where $1<p<p_1<2$, $\frac{1}{p}=\frac{1}{p_1}+\frac{1}{p_2}$.

 For $J_2^\vep$, applying Minkowski inequality, H\"{o}lder's inequality and Sobolev
 inequality, as well as
  estimates for Poisson equation, noticing $\psi^\vep\in L^\infty([0,\infty);H^1(\Bbb R^3))$
  and Lemma \ref{redlem1}, we estimate
 \beas
 \|J_2^\vep\|_p\leq\|P_\vep(\varphi^\vep)(\psi^\vep-\Pi\psi^\vep)w_0\|_{(1,p)}\leq C
 \|P_\vep(\varphi^\vep)\|_{L^{p^\ast}(\Bbb R^3)}\|\psi^\vep-\Pi\psi^\vep\|_{(\infty,2)}
 \leq C\vep\|\frac{|\psi^\vep|^2}{\vep}\|_{L^{p^\ast}(\Bbb R^3)} \leq C,
 \eeas
 where $p^\ast=2p/(2-p)\leq 3$.

 For $J^\vep_3$, similar as $J^1_\vep$, $J^2_\vep$, we have
 \beas
 \|J_3^\vep\|_p&\leq&\|P_\vep(\varphi_2^\vep)\Pi\psi^\vep w_0\|_{(1,p)}\leq C\|P_{\vep}
 (\varphi_2^\vep)\|_{L^{p_1}(\Bbb R^3)}\|\phi^\vep\|_{L^{p_2}(\Bbb R^2)}
 \leq\frac{C}{\vep}\|\,|\psi^\vep|^2-|\Pi\psi^\vep|^2\,\|_{L^{p_1}(\Bbb R^3)}\\&\leq&\frac{C}
 {\vep}\|\psi^\vep-\Pi\psi^\vep\|_{L^2(\Bbb R^3)}(\|\psi^\vep\|_{L^{p_3}(\Bbb R^3)}+
 \|\Pi\psi^\vep\|_{L^{p_3}(\Bbb R^3)})\leq C,
 \eeas
 where $p_3=2p_1^2/(2-p_1)\leq 6$. Hence, by choosing $p=6/5$, and $p_1=4/3$ would
 satisfy all the conditions for $J_k^\vep$ (k=1,2,3), where we shall derive that uniformly in $t$,
 \be
 \|g^\vep\|_{L^p(\Bbb R^2)}\leq \|J_1^\vep\|_{L^p(\Bbb R^2)}+\|J_2^\vep\|_{L^p(\Bbb R^2)}
 +\|J_3^\vep\|_{L^p(\Bbb R^2)}\leq C.
 \ee
  Then choose $q=p$ in (\ref{red:int}), we have
 \be
 \|\chi^\vep\|_{L^\infty([0,T];L^2(\Bbb R^2))}\leq C
 \left[\|\chi^\vep\|_{L^{\rho^\ast}([0,T];L^2(\Bbb R^2))}+\vep\right].
 \ee
 Applying the results for all $t\in[0,T]$, we find
 \be
 \|\chi^\vep(t)\|_2^{\rho^\ast}\leq C\left[\int_0^t\|\chi^\vep(s)\|_2^{\rho^\ast}\,ds
 +\vep^{\rho^\ast}\right],\quad t\in[0,T],
 \ee
 and Gronwal's inequality will give that $\|\chi^\vep(t)\|_2\leq C\vep$ for all $t\in[0,T]$.
 Combining with (\ref{L2error}), we can draw the desired conclusion.
\hfill$\Box$

 \subsection{Reduction to 1D}
 In this case, we again consider the weak interaction regime $\beta\to \vep^{2}\beta$,
 $\lambda\to\vep^{2}\lambda$.
In {\sl{Case II}} (\ref{case2}),  for the full 3D GPPS
(\ref{gpe})-(\ref{poisson}), introducing the re-scaling $x\to
\vep x$, $y\to \vep y$, $\psi\to \vep^{-1}\psi^\vep$ which preserves the
normalization,  then
 \be\label{rescalgpe1}
 i\p_t\psi^\vep(\bx,z,t)=\left[\H_z^V+\frac{1}{\vep^2}\H_{\bx}+(\beta-\lambda)|\psi^\vep|^2-3
 \vep\lambda\p_{\tilde{\bn}_{\vep}\tilde{\bn}_\vep}\varphi^\vep \right]\psi^\vep,\quad (\bx,z)\in \Bbb R^3,
 \ee
 where $\bx=(x,y)\in\Bbb R^2$ and
 \bea
 &&\H_z^V=-\frac{1}{2}\p_{zz}+V_1(z),\quad \H_{\bx}=-\frac12(\p_{xx}+\p_{yy}+x^2+y^2),\\
 &&\tilde{\bn}_\vep=(n_1/\vep,n_2/\vep,n_3),\quad \p_{\tilde{\bn}_\vep}=\tilde{\bn}_\vep
 \cdot\nabla,\quad \p_{\tilde{\bn}_\vep\tilde{\bn}_\vep}=\p_{\tilde{\bn}_\vep}(\p_{\tilde{\bn}_\vep}),\\
 &&(-\frac{1}{\vep^2}\p_{xx}-\frac{1}{\vep^2}\p_{yy}-\p_{zz})\varphi^\vep=\frac{1}{\vep^2}|\psi^\vep|^2,\qquad
 \lim\limits_{|\bx|\to \infty}\varphi^\vep(\bx)=0.\label{rescalps1}
 \eea
 Note that the ground state mode of $\H_{\bx}$ would be given by $w_0(x)w_0(y)$ with eigenvalue 1,
 and the initial data is then assumed to be
\be\label{init1}
 \psi^\vep(\bx,z,0)=\phi_0(z)w_0(x)w_0(y),\quad \phi_0\in X_1 \mbox{ and } \|\phi_0\|_{L^2(\Bbb R)}=1.
 \ee

In {\sl{Case II}} (\ref{case2}) in the introduction,
 when $\vep\to0^+$, the quasi-1D equation  (\ref{gpe1d}) will lead to an $\vep$-independent equation in the weak interaction regime,
 \be\label{eq:1d:weak}
 i\p_t \phi(z,t)=\H_z^V\phi+\frac{\beta+\frac12\lambda (1-3n_3^2)}{2\pi}|\phi|^2\phi,\qquad z\in\Bbb
 R,\quad t>0,
 \ee
 with the initial condition $\phi(z,0)=\phi_0(z)$.

 Following the steps in the last subsection, we can prove the following results.
 \begin{theorem}(Dimension reduction to 1D)
 Suppose the real-valued trap
potential satisfies  $V_1(z)\ge0$ for $z\in{\Bbb R}$ and $V_1(z)\in
C^\infty(\Bbb R)$, $D^k V_1(z)\in L^\infty(\Bbb R)$ for all
$k\ge 2$. Assume $-\frac{\beta}{2}\leq \lambda \leq\beta$ and
$\beta\ge0$, and let $\psi^\vep\in C([0,\infty);X_3)$ and $\phi\in
C([0,\infty);X_1)$ be the unique solutions of the equations
(\ref{rescalgpe1})-(\ref{init1}) and (\ref{eq:1d:weak}),
respectively, then for any $T>0$, there exists $C_{T}>0$ such that
\be
\left\|\psi^\vep(\bx,z,t)-e^{-it/\vep^2}\phi(z,t)w_0(x)w_0(y)\right\|_{L^2(\Bbb
R^3)}\leq C_{T}\,\vep, \qquad \forall t\in [0,T]. \ee
\end{theorem}

\section{Conclusion}
\setcounter{equation}{0}

We have analyzed the effective lower dimensional models  for three
dimensional Gross-Pitaevskii-Poisson system (GPPS) describing
dipolar Bose-Einstein condensates (BEC)  in anisotropic confinement.
The
 quasi-2D approximate equations I (\ref{gpe2d})
and  II (\ref{gpe2d2})  are introduced in the case that the trap is
strongly confined in the vertical $z$-direction, and  the quasi-1D
approximate equation (\ref{gpe1d}) is presented in the case that the
trap is strongly confined in the $x$-, $y$- directions. Properties
of ground states for all equations, such as existence and uniqueness
as well as non-uniqueness results were studied. Well-posedness  of
the Cauchy problems for both equations and possible finite time
blow-up in 2D case were discussed.
 Finally, we
rigorously proved the linear convergence rate of the dimension
reduction
 from 3D GPPS to its quasi-2D and quasi-1D approximations in the weak interaction
regime, i.e. $\beta=\lambda=O(\vep^{3-d})$, in lower $d$ ($d=1,2$)
dimensions. In such situation, all the nonlocal terms in the
effective equation (\ref{gpe2d}), (\ref{gpe2d2}) and (\ref{gpe1d})
vanish, resulting in a regular GPE in lower dimensions. We remark
that the results in the paper hold true for a larger class of
confinements rather that the harmonic ones. In fact,
effective two-dimensional models
have been derived and analyzed recently for a
multilayer stack of dipolar BEC formed by a
strong lattice potential \cite{MCB}.


\bigskip

\noindent{\bf Acknowledgment} This work was supported in part by the
Academic Research Fund of Ministry of Education of Singapore grant
R-146-000-120-112 (W. B. and Y.C.). This work was partially done
while W.B. and Y.C. were visiting the Institut de Math\'{e}matiques
de Toulouse at Universit\'{e} Paul Sabatier  in 2010. Y.C. would like to thank the
support   from  the European Union programme ``Differential Equations
with Applications in Science and Engineering" MEST-CT-2005-021122 during his visit.

\end{document}